\documentclass[final]{siamltex}
\usepackage{amsmath, amssymb, amsfonts}
\usepackage{graphicx,color}
\usepackage{diagbox}
\usepackage[ruled]{algorithm2e}
\usepackage{algorithmic}
\usepackage{enumitem}
\usepackage{multirow}
\usepackage{subcaption}
\usepackage{comment}
\usepackage{epstopdf}
\usepackage{url}
\usepackage{hhline}
\usepackage{booktabs}
\usepackage[table]{xcolor}

\newtheorem{myrem}{Remark}

\newcommand{\bvec}[1]{\mathbf{#1}}

\renewcommand{\Re}{\mathrm{Re}~}
\renewcommand{\Im}{\mathrm{Im}~}
\newcommand{\Tr}{\mathrm{Tr}}

\newcommand{\mc}[1]{\mathcal{#1}}

\newcommand{\va}{\bvec{a}}
\newcommand{\vb}{\bvec{b}}

\newcommand{\vr}{\bvec{r}}
\newcommand{\vk}{\bvec{k}}
\newcommand{\vK}{\bvec{K}}
\newcommand{\vR}{\bvec{R}}

\newcommand{\ud}{\,\mathrm{d}}

\newcommand{\abs}[1]{\lvert#1\rvert}

\newcommand{\bra}[1]{\langle#1\rvert}
\newcommand{\ket}[1]{\lvert#1\rangle}
\newcommand{\wt}[1]{\widetilde{#1}}

\newcommand{\ie}{\emph{i.e.}}
\newcommand{\wannier}{\texttt{Wannier90} }
\newcommand{\wannierp}{\texttt{Wannier90}}
\newcommand{\qe}{\texttt{Quantum ESPRESSO}}

\newcommand{\RR}{\mathbb{R}}
\newcommand{\CC}{\mathbb{C}}

\newcommand{\I}{\imath}

\newcommand{\R}{\mathbb{R}}
\newcommand{\C}{\mathbb{C}}

\newcommand{\N}{\mathbb{N}}
\newcommand{\lela}{\left \langle}
  \newcommand{\rira}{\right \rangle}

\title{Variational Formulation for \\Wannier functions With
Entangled Band Structure}
\author{Anil Damle \thanks{Department of Computer Science, Cornell
    University, Ithaca, NY 14853. Email: \texttt{damle@cornell.edu}}
  \and Antoine Levitt \thanks{Inria Paris, F-75589 Paris Cedex 12,
    Universit′e Paris-Est, CERMICS (ENPC), F-77455 Marne-la-Vall\'ee.
    Email: \texttt{antoine.levitt@inria.fr}}  \and Lin Lin \thanks{Department of Mathematics,
University of California, Berkeley, Berkeley, CA 94720 and Computational
Research Division, Lawrence Berkeley National Laboratory, Berkeley, CA
94720. Email: \texttt{linlin@math.berkeley.edu}}}

\begin{document}

\maketitle

\begin{abstract}
  Wannier functions provide a localized representation of spectral
  subspaces of periodic Hamiltonians, and play an important role for
  interpreting and accelerating Hartree-Fock and Kohn-Sham density
  functional theory calculations in quantum physics and chemistry. For
  systems with isolated band structure, the existence of exponentially
  localized Wannier functions and numerical algorithms for finding
  them are well studied. In contrast, for systems with entangled band
  structure, Wannier functions must be generalized to span a subspace
  larger than the spectral subspace of interest to achieve favorable spatial
  locality. In this setting, little is known about the theoretical
  properties of these Wannier functions, and few algorithms can find
  them robustly. We develop a variational formulation to compute these
  generalized maximally localized Wannier functions. When
  paired with an initial guess based on the selected columns of the
  density matrix (SCDM) method, our method can robustly find Wannier
  functions for systems with entangled band structure. We formulate
  the problem as a constrained nonlinear optimization problem, and
  show how the widely used disentanglement procedure can be
  interpreted as a splitting method to approximately solve this
  problem. We demonstrate the performance of our method using real materials including
  silicon, copper, and aluminum. To examine more precisely the
  localization properties of Wannier functions, we study the free
  electron gas in one and two dimensions, where we show that the
  maximally-localized Wannier functions only decay algebraically. We
  also explain using a one dimensional example how to modify them to obtain
  super-algebraic decay.

\end{abstract}


\pagestyle{myheadings}
\thispagestyle{plain}

\section{Introduction}

Localized representations of electronic wavefunctions have a wide range
of applications in quantum physics, chemistry, and materials science.
In an effective single particle theory such as Hartree-Fock theory
and Kohn-Sham density functional theory
(KSDFT)~\cite{HohenbergKohn1964,KohnSham1965}, the electronic
wavefunctions are characterized by eigenfunctions of single particle
Hamiltonian operators. These eigenfunctions generally have significant
magnitude in large portions of the computational domain. However, the physically meaningful quantity is
not each individual eigenfunction, but the subspace spanned by the
collection of a set of eigenfunctions. This is often referred to as the Kohn-Sham subspace, and it is often possible to reduce the computational complexity of various methods by using an alternative, localized representation of the subspace.

Wannier functions provide one such localized representation of the Kohn-Sham
subspace. They require
significantly less memory to store, and are the foundation of
so-called ``linear scaling
methods''~\cite{Kohn1996,Goedecker1999,BowlerMiyazaki2012}
for solving quantum problems.  They can also be used to analyze chemical
bonding in complex materials, interpolate the band structure of
crystals, accelerate ground and excited state electronic structure
calculations, and form reduced order models for strongly correlated
many body systems~\cite{MarzariMostofiYatesEtAl2012}.

Wannier functions are not uniquely determined, and depend on a
choice of gauge (a rotation among the occupied states), which strongly
influences their localization. For periodic systems with an isolated
band structure, the localization properties of Wannier functions have
been studied extensively
\cite{Kohn1959,Blount1962,Nenciu1991,BrouderPanatiCalandraEtAl2007,PanatiPisante2013}.
Interestingly, the existence of localized Wannier functions in this
case is characterized by a topological invariant. For physical systems
without magnetic field (invariant under time-reversal symmetry), this
topological invariant is trivial, and it is known that there exists a
gauge leading to Wannier functions with exponential decay
\cite{BrouderPanatiCalandraEtAl2007,panati2007triviality}. In this
setting, efficient numerical algorithms have been developed to compute
these exponentially localized functions
\cite{MarzariVanderbilt1997,KochGoedecker2001,Gygi2009,ELiLu2010,OzolinsLaiCaflischEtAl2013,DamleLinYing2015,MustafaCohCohenEtAl2015,CancesLevittPanatiEtAl2017,DamleLinYing2017a}.
In particular, the widely used maximally-localized Wannier function
(MLWF) procedure minimizes the variance (or
``spread'')~\cite{FosterBoys1960,MarzariVanderbilt1997} over all
possible choices of gauge to obtain localized Wannier functions. In
the insulating case, it is known that minimizers of this spread are
exponentially localized \cite{PanatiPisante2013}.

The situation becomes significantly more challenging for systems with
entangled band structure. Entangled band structure arises in metallic
systems, but also in insulating systems when conduction bands or a
selected range of bands are to be localized. A straightforward
definition of Wannier functions requires the set of all Wannier
functions to exactly span the selected spectral subspace. However,
such Wannier functions are known to decay slowly in real space.
Therefore, the definition of Wannier functions has been generalized to
refer to functions spanning a subspace larger than, but containing the given
entangled spectral subspace, referred to as a ``frozen window''
\cite{SouzaMarzariVanderbilt2001}. This is useful for instance in band
interpolation, where the additional Wannier functions give rise to
extra bands that can simply be ignored. Finding such generalized Wannier
functions numerically is considerably more complex, and few algorithms
in the literature accomplish this task in a robust fashion.
Furthermore, little is known theoretically about the localization
properties of the constructed generalized Wannier functions. In order
to be consistent with the terminology in the physics literature, we
will refer to these generalized Wannier functions simply as Wannier
functions, unless otherwise noted.

In this paper, we develop a variational formulation for finding Wannier
functions in the entangled setting. We formulate the problem as a nonlinear constrained
optimization problem. Practical Wannier function calculations indicate
that such nonlinear optimization problems can have many local minima.
Hence the solution can strongly depend on the initial guess, and the difficulty of constructing a good initial guess is often a significant impediment to
finding Wannier functions in a robust fashion. In order to avoid
being trapped at undesirable local minima, we use the recently
developed selected columns of the density matrix (SCDM) methodology to
construct the initial guess for our variational formulation. This
strategy is applicable to both the isolated
case~\cite{DamleLinYing2015} and the entangled
case~\cite{DamleLin2017}.

Our variational formulation can be obtained in several theoretically
equivalent constructions.  We find that one of these formulations yields
the so-called partly occupied Wannier
functions~\cite{ThygesenHansenJacobsen2005}, and can be solved efficiently using
standard numerical algorithms for minimization under orthogonality constraints.
Our formulation also reveals that the widely used ``disentanglement''
procedure~\cite{SouzaMarzariVanderbilt2001} can be viewed as a splitting
method for solving the constrained optimization problem, which only
performs a single alternation step between the two pieces of the
objective function, and therefore does not achieve a global minimum of
the spread. We verify the performance of the variational
formulation with real materials such as silicon, copper, and aluminum.
In these examples, we find that the fully converged variational
formulation consistently provides orbitals with a smaller spread than
that from the disentanglement procedure, and is more robust to the
choice of initial guess. We also find that the difference between the
variational formulation and the disentanglement procedure is often
small when used for band structure interpolation.
%

The variational formulation allows us to study the decay properties of
Wannier functions for metallic systems. We present the localization
properties of generalized Wannier functions for the free electron gas
in one and two dimensions. We find that minimizers of the spread
exhibit a weak algebraic decay, related to singularities that we
identify in $\vk$ space. This slow decay is shown to be not a
fundamental property of disentangled Wannier functions, but rather a
consequence of the fact that minimizing the spread only imposes finite
second moments (or square-integrable first derivatives in $\vk$
space). In particular we show in one dimension how to modify the
maximally-localized Wannier functions to obtain super-algebraic decay.

The rest of the paper is organized as follows.  We first introduce several background topics such as Bloch-Floquet theory, Wannier functions,
and the SCDM methodology in
section~\ref{sec:prelim}. We then present a variational formulation for
Wannier functions in section~\ref{sec:varwannier}, and discuss the
relation between our variational formulation and existing methods.
Numerical results for real materials and for the free electron gas are
given in sections~\ref{sec:materials} and ~\ref{sec:electrongas}, followed by conclusion and discussion
in section~\ref{sec:conclusion}. Some of the technical details related
to the implementation of the variational formulation are given in the Appendix~\ref{app:gradient}.

\section{Preliminaries}\label{sec:prelim}

\subsection{Bloch-Floquet theory}

We first briefly review
Bloch-Floquet theory for crystal structures. 
The \emph{Bravais lattice} with lattice vectors $\va_{1},\va_{2},\va_{3}\in \RR^{3}$ is defined as
\begin{equation}
  \mathbb{L} = \left\{ \vR = \sum_{i=1}^{3} n_{i} \bvec{a}_{i}, \quad n_1,n_2,n_3\in\mathbb{Z}\right\},
  \label{eqn:Blattice}
\end{equation}
and the lattice vectors define a unit cell in the Bravais lattice
\begin{equation}
 \Gamma =\left\{ \vr = \sum_{i=1}^{3} c_{i} \bvec{a}_{i} \vert -1/2 \le c_{1},c_{2},c_{3}< 1/2\right\}.
 \label{eqn:unit_cell}
\end{equation}
The Bravais lattice induces a reciprocal lattice denoted
$\mathbb{L}^{*}$, which is the support of the Fourier transform of
$\mathbb L$-periodic functions. The lattice vectors of $\mathbb{L}^{*}$
are denoted by
$\vb_{1},\vb_{2},\vb_{3}$, with $\vb_{i}\cdot\va_{j} = 2\pi \delta_{ij}$. A unit cell of the
reciprocal lattice is selected and called (with some abuse of
language) the \emph{Brillouin zone}, and is defined as
\begin{equation}
  \Gamma^{*} =\left\{ \vk = \sum_{i=1}^{3} c_{i} \bvec{b}_{i} \vert -1/2 \le c_{1},c_{2},c_{3}< 1/2\right\}.
 \label{eqn:rec_unit_cell}
\end{equation}

For a potential $V$ that is real-valued and $\mathbb{L}$-periodic, \ie~ 
\begin{equation}
V\left( \vr+\vR\right) = V(\vr), \quad \forall \vr\in \RR^3, \vR \in \mathbb{L},
\label{eqn:Vperiod}
\end{equation} 
we consider the Schr\"{o}dinger operator in $\mathbb{R}^3$
\[
\mc{H} = -\frac{1}{2}\Delta + V.
\]

The Bloch-Floquet theory allows us to relabel the spectrum of $\mc{H}$ using two indices $(n,\vk)$, where $n\in \mathbb{N}$ is the band
index, and $\vk\in\Gamma^{*}$ is the Brillouin zone index. The
generalized (not square-integrable) eigenfunction $\psi_{n,\vk}(\vr)$ is known as a Bloch orbital and 
satisfies
\[
\mc{H}\psi_{n,\vk}(\vr)=\varepsilon_{n,\vk}\psi_{n,\vk}(\vr).
\]
Importantly, $\psi_{n,\vk}$ can be decomposed as 
\begin{equation}
\psi_{n,\vk}(\vr) = e^{\I \vk\cdot \vr} u_{n,\vk}(\vr),
  \label{}
\end{equation}
where $u_{n,\vk}(\vr)$ is a periodic function with respect to
$\mathbb{L}$. Eigenpairs $(\varepsilon_{n,\vk},u_{n,\vk})$ can therefore be 
obtained by solving the eigenvalue problem
\begin{equation}
  \mc{H}(\vk) u_{n,\vk} = \varepsilon_{n,\vk} u_{n,\vk}, \quad n \in \mathbb N, \quad \vk \in
  \Gamma^{*},
  \label{eqn:bandproblem}
\end{equation}
where $\mc{H}(\vk) = \frac12 (-\I \nabla+ \vk)^2 + V(\vr)$.
For each $\vk$, the eigenvalues $\varepsilon_{n,\vk}$ are ordered
non-decreasingly, and $\{\varepsilon_{n,\vk}\}$ as a
function of $\vk$ for a fixed $n$ is called a
\emph{band}. The set of all eigenvalues is called the
\emph{band structure} of the crystal and characterizes the spectrum
of the operator $\mc{H}$. 
If $\varepsilon_{N+1,\vk} > \varepsilon_{N,\vk}$ for all $\vk \in \Gamma^{*}$,
then the first $N$ bands are \textit{isolated}. This is for instance the case
in the occupied bands of an insulator. 
When the gap condition $\varepsilon_{N+1,\vk} > \varepsilon_{N,\vk}$
does not hold, the band
structure becomes entangled. Entangled band structure appears not only
in metallic systems, but also insulating systems when a Wannier
representation of part of the conduction bands is required.


\subsection{Wannier functions}
\label{sec:wannier}

For simplicity, we first consider systems with isolated first $N$ bands\textemdash an assumption we will drop towards the end of this section. 
Rotating the set of functions $\{\psi_{n,\vk}\}$ by an arbitrary
unitary matrix $U(\vk)\in\CC^{N\times N}$, we can
define a new set of functions
\begin{equation}
  \wt{\psi}_{n,\vk}(\vr) = \sum_{m=1}^{N} \psi_{m,\vk}(\vr) U_{m,n}(\vk), \quad
  \vk\in \Gamma^{*}.
  \label{eqn:Utransform}
\end{equation}
A given set of of such matrices $\{U(\vk)\}_{\vk \in \Gamma^{*}}$ is called a \emph{gauge}.

For each $\vk$,
we consider the density matrix $P(\vk)$, which is the projector on the
the eigenspace corresponding to the first $N$ eigenvalues of $H(\vk)$
\begin{equation}
  P(\vk) = \sum_{n=1}^{N}
  \ket{\psi_{n,\vk}}\bra{\psi_{n,\vk}} =\sum_{n=1}^{N}
  \ket{\wt{\psi}_{n,\vk}}\bra{\wt{\psi}_{n,\vk}}.
  \label{eqn:projector}
\end{equation}
Importantly, for each $\vk$, the density matrix $P(\vk)$ is
gauge-invariant. If $\mathcal C$ is a contour in the complex plane enclosing the eigenvalues
$\varepsilon_{1,\vk},\dots, \varepsilon_{N,\vk}$ (and only those), then the Cauchy
integral formula yields an alternative representation of $P(\vk)$
\begin{equation}
  P(\vk) = \frac 1 {2\pi \I}\int_{\mathcal C} \frac 1 {\lambda - H(\vk)}\ud \lambda.
\end{equation}
Since $H(\vk)$ is analytic, it
follows that so is $P(\vk)$.

Given a choice of gauge, the Wannier functions are defined as~\cite{Wannier1937}
\begin{equation}
  w_{n,\vR}(\vr) = \frac{1}{\lvert\Gamma^{*}\rvert}\int_{\Gamma^{*}} \wt{\psi}_{n,\vk}(\vr)
  e^{-\I \vk\cdot \vR}
  \ud \vk, \quad \vr\in\RR^3, \vR\in \mathbb{L},
  \label{eqn:wannier}
\end{equation}
where $\lvert\Gamma^{*}\rvert$ is the volume of the first Brillouin zone. This represents a unitary transformation from the family
$(\psi_{n,\vk})_{n=1,\dots,N,\vk \in \Gamma^{*}}$ to
$(w_{n,\vR})_{n=1,\dots,N,\vR \in \mathbb L}$. In particular, the
Wannier functions $w_{n,\vR}$ are orthogonal to each other and span
the same space as the range of the total density matrix
$\frac 1 {|\Gamma^{*}|}\int_{\Gamma^{*}} P(\vk)\ud \vk$. They are also
translation invariant: $w_{n,\vR}(\vr) = w_{n}(\vr-\vR)$.

For insulating systems, in the absence of topological obstructions,
there exists a gauge such that $\wt{\psi}_{n,\vk}$ is analytic and
$\mathbb{L}^{*}$-periodic in $\vk$, implying that the Fourier
transform of $w_{n,\vR}$ is analytic, and therefore that each
Wannier function decays exponentially as
$\abs{\vr}\to \infty$~\cite{Blount1962,BrouderPanatiCalandraEtAl2007,panati2007triviality}.
The Wannier localization problem is reduced to the problem
of finding a gauge $\{U(\vk)\}$ such that $w_{n,\bvec{0}}$ is localized, or, equivalently, that $\wt{\psi}_{n,\vk}$ is smooth with respect
to $\vk$. This can be done by minimizing the ``spread
functional''~\cite{FosterBoys1960,MarzariVanderbilt1997}
\begin{equation}
  \Omega[\{U(\vk)\}] = \sum_{n=1}^{N} \int \abs{w_{n,\bvec{0}}(\vr)}^{2} \vr^2 \ud \vr - \left| 
   \int \abs{w_{n,\bvec{0}}(\vr)}^{2} \vr \ud \vr \right|^2.
   \label{eqn:spread}
\end{equation}
Here $w_{n,\bvec{0}}$ depends on $U(\vk)$ through $\wt{\psi}_{n,\vk}$ as
in Eq.~\eqref{eqn:wannier}.  This problem is usually
solved by a minimization algorithm such as steepest descent or
conjugate gradient with projections at each step to respect the constraints
that $U(\vk)$ must be unitary \cite{MostofiYatesLeeEtAl2008,mostofi2014updated}.


For systems with entangled band structure, the density matrix $P(\vk)$
as defined in Eq.~\eqref{eqn:projector} is no longer smooth with
respect to $\vk$. As a result, there is no choice of gauge $U(\vk)$
that leads to a set of rotated Bloch orbitals that is smooth with
respect to $\vk$, and Wannier functions defined strictly according to
Eq.~\eqref{eqn:wannier} will then decay very slowly in real
space~\cite{Goedecker1999}. In order to enhance the localization
properties of Wannier functions, the definition of Wannier functions
has been generalized so that the spectral subspace interest is only a
proper subspace of that spanned by Wannier
functions~\cite{SouzaMarzariVanderbilt2001}. In the physics literature,
the spectral subspace is described by a ``frozen window'' along the
energy spectrum, and the Wannier functions are linear combinations of
orbitals from a larger set described by an ``outer window''.


More specifically, we first fix a number of bands $N_{o}$ that
determines the outer
window\footnote{To simplify the exposition we assume a constant number of bands in the outer window, but this can be relaxed to a variable number of bands $N_{o}(\vk)$.} and then proceed to look for Wannier functions
built out of $\psi_{1,\vk},\dots,\psi_{N_{o},\vk}$. Next, for each $\vk$ point we fix a
set of frozen bands $\mathcal{N}_{f}(\vk) \subset \left[N_{o}\right],$ and let $N_f(\vk) = \lvert \mathcal{N}_f(\vk)\rvert.$ $\mathcal{N}_f(\vk)$ are often defined as the bands within a fixed energy window that we
will try to reproduce. Correspondingly, we
define the frozen density matrix as 
\begin{equation}
  P_{f}(\vk) = \sum_{n\in \mathcal{N}_{f}(\vk)}
  \ket{\psi_{n,\vk}}\bra{\psi_{n,\vk}},
  \label{eqn:projectorfrozen}
\end{equation}
which is the projection onto the states within the frozen energy
window. Again, the frozen density
matrix as defined in Eq.~\eqref{eqn:projectorfrozen} is not smooth with
respect to $\vk$.


We now seek to construct a set of $N_{w}$ Wannier functions
that span the subspace defined by the range of
$\frac 1 {|\Gamma^{*}|}\int_{\Gamma^{*}} P_{f}(\vk) \ud \vk$. We introduce
the gauge matrices $U(\vk)\in \CC^{N_{o}\times N_{w}}$ with orthonormal
columns, with $\left\lvert N_{f}(\vk)\right\rvert \leq N_{w} \leq N_{o}$, such that
\begin{equation}
  \wt{\psi}_{n,\vk}(\vr) = \sum_{m=1}^{N_{o}} \psi_{m,\vk}(\vr) U_{m,n}(\vk), \quad
  \vk\in \Gamma^{*},  n=1,\ldots,N_{w}.
  \label{eqn:rotateentangle}
\end{equation}
This may equivalently expressed in matrix form as $\wt{\Psi}(\vk) = \Psi(\vk) U(\vk)$.
This choice of gauge also induces a density matrix of rank $N_{w}$ for each $\vk$ defined as
\begin{equation}
  P_{w}(\vk) = \sum_{n=1}^{N_{w}}
  \ket{\wt{\psi}_{n,\vk}}\bra{\wt{\psi}_{n,\vk}} = \wt{\Psi}(\vk)
  \wt{\Psi}^{*}(\vk) = \Psi(\vk) U(\vk) U^{*}(\vk) \Psi^{*}(\vk).    
  \label{eqn:projectorentangle}
\end{equation}
Note that unlike the case with isolated band structure where $N_{o}=N_{w},$ here $N_w\neq N_o$ implies that $U(\vk)
U^{*}(\vk) \neq I_{N_{o}}.$  Furthermore, since the set of orbitals in the frozen window is only a subset of all
possible orbitals, in general the projectors $P_{w}$ and $P_{f}$ do
not span the same space. In
order to ensure that our Wannier functions span the same subspace as the
subspace associated with the frozen window, we require that
\begin{equation}
  P_{w}(\vk) P_{f}(\vk) = P_{f}(\vk), \quad \forall
  \vk\in\Gamma^{*}.
  \label{eqn:projectorcondition}
\end{equation}

\section{Variational formulation for Wannier functions with entangled
band structure}\label{sec:varwannier}

We now proceed to develop a variational formulation for Wannier
functions. First, we illustrate how to encode the desired constraints
when paired with the aforementioned spread functional. Subsequently, to facilitate numerical solution of the optimization problem, we refine how the constraints are expressed. Lastly, we discuss the relation the existing disentanglement procedure to our formulation and discuss how we construct an initial guess using the SCDM methodology.

\subsection{Formulating the optimization problem}
Without loss of generality, for the following discussion we assume that the
frozen orbitals ($\Psi_{f}$) are always ordered before the rest of the
orbitals ($\Psi_{r}$). In terms of the notation from the previous
section, this means that for each $\vk$ the frozen orbitals are
represented by the set $\mathcal{N}_f(\vk) = \left\{1,2,\ldots,N_f(\vk)\right\}$ with
$N_f(\vk)$ simply representing the number of frozen orbitals per
$\vk$-point. Now, we may partition the orbitals and the gauge using
the following block form
\begin{equation}
  \Psi(\vk)  = \begin{bmatrix}\Psi_{f}(\vk)& \Psi_{r}(\vk)\end{bmatrix},
   \quad U(\vk) =
  \begin{bmatrix} U_{f}(\vk) \\ U_{r}(\vk)\end{bmatrix}.
\end{equation}

The matrices $U_{f}(\vk)$ and $U_{r}(\vk)$ are of size
$N_{f}(\vk) \times N_{w}$ and $(N_{o} - N_{f}(\vk)) \times N_{w}$, 
encoding the weight assigned to the frozen
subspace and its complement in the Wannier functions, respectively. The condition that the Wannier functions represent the
frozen bands as in Eq.~\eqref{eqn:projectorcondition} can conveniently
be expressed in terms of these matrices as follows.
\begin{proposition}\label{prop:projection}
  The following statements are equivalent:
  \begin{enumerate}
    \item $P_{w}(\vk) P_{f}(\vk) = P_{f}(\vk)$.
    \item $U_{f}(\vk) U^{*}_{f}(\vk) = I_{N_{f}(\vk)}$.
    \item $U_{f}(\vk) U^{*}_{r}(\vk) = 0$ and $U_{f}(\vk)$ has full row
      rank.
    \item $U(\vk)=\begin{bmatrix}
    I_{N_{f}(\vk)} & 0 \\
    0 & Y(\vk)
  \end{bmatrix} X(\vk)$,  where $X(\vk)$ is a unitary matrix of size $N_{w} \times N_{w}$,
  and  $Y(\vk)$ is a matrix with orthogonal columns of size
  $(N_{o}-N_{f}(\vk)) \times (N_{w}-N_{f}(\vk))$.
  \end{enumerate}
\end{proposition}
\begin{proof}
  Since each $\vk$ point is treated independently, for simplicity we
  drop the $\vk$ dependence in the proof below.

   \noindent$\mathbf{1 \Leftrightarrow 2}:$ From the definition of
   $P_{w},P_{f}$ we have
   \begin{align*}
     P_{w}P_{f} &= \Psi      \begin{bmatrix}
       U_{f} U_{f}^{*} & U_{f} U_{r}^{*}\\
       U_{r}U_{f}^{*} & U_{r}U_{r}^{*}
     \end{bmatrix}
     \begin{bmatrix}
       I&0\\0&0
     \end{bmatrix}
     \Psi^{*} = \Psi      \begin{bmatrix}
       U_{f} U_{f}^{*} & 0\\
       U_{r}U_{f}^{*} & 0
     \end{bmatrix}
     \Psi^{*}\\
     P_{f} &= \Psi
     \begin{bmatrix}
       I_{N_{f}}&0\\0&0
     \end{bmatrix}
     \Psi^{*}
   \end{align*}
   The result follows since $\Psi$ has orthogonal columns.

   \noindent$\mathbf{2 \Leftrightarrow 3}:$
     From the partition of $U$ we have
     \begin{align*}
       U U^{*} =
       \begin{bmatrix}
         U_{f} U_{f}^{*} & U_{f} U_{r}^{*}\\
         U_{r}U_{f}^{*} & U_{r}U_{r}^{*}
       \end{bmatrix} 
     \end{align*}
   If $U_{f} U_{f}^{*} = I_{N_{f}}$, since $UU^{*}$ is a projector, it
   follows that $U_{r}U_{f}^{*} = 0$. On the other hand, if
   $U_{r}U_{f}^{*} = U_{f}U_{r}^{*} = 0$, then the fact that $UU^{*}$ is a
   projector implies that $U_{f}U_{f}^{*}$ is a projector as well. Since
   it has full row rank, it must therefore be the identity matrix.

   \noindent$\mathbf{3 \Leftrightarrow 4}:$ 
   If 3 is true, then
   \[
     P_{U} = U U^{*} = 
     \begin{bmatrix}
       I_{N_{f}} & 0 \\
       0 & U_{r}U_{r}^{*}
     \end{bmatrix}.
   \]
   Then $U_{r}U_{r}^{*}$ is a projector with rank $N_{w}-N_{f}$, or equivalently
   \[
     U_{r}U_{r}^{*} = Y Y^{*}
   \]
   for some $(N_{o}-N_{f}) \times (N_{w}-N_{f})$ matrix $Y$ with
   orthogonal columns. Then
   \[
     U = P_{U} U = 
     \begin{bmatrix}
       I_{N_{f}} & 0 \\
       0 & Y Y^{*} 
     \end{bmatrix} U = 
     \begin{bmatrix}
       I_{N_{f}} & 0 \\
       0 & Y  
     \end{bmatrix} X, 
   \]
   where 
   \[
     X=\begin{bmatrix} I_{N_{f}} & 0 \\ 0 & Y^{*} \end{bmatrix} U,
   \]
   and it can be readily verified that $X$ is unitary. The reverse
   direction is obvious.
\end{proof}

Proposition~\ref{prop:projection} gives us
various concise ways to impose the desired condition on the span of
Wannier functions, and we may for instance consider using condition $2$. Since the smoothness requirement for $\wt{\psi}_{n,\vk}$ with respect to the
Brillouin zone index $\vk$ can be realized by minimizing the spread
functional~\eqref{eqn:spread}, finding the desired smooth gauge $U(\vk)$
can be recast as the following constrained optimization problem:

\begin{equation}
  \begin{split}
  \inf_{\{U(\vk)\}} &\quad \Omega[\{U(\vk)\}]\\
  \text{s.t.} & \quad U^{*}(\vk) U(\vk) = I_{N_{w}}, \quad
  U_{f}(\vk) U_{f}^{*}(\vk) = I_{N_{f}(\vk)}.
  \end{split}
  \label{eqn:varopt}
\end{equation}

The difficulty at this stage is that numerical optimization of \eqref{eqn:varopt} with
respect to these constraints may not be easy. In particular, the set of
matrices $U$ satisfying $U^{*}U = I_{N_{w}}$ and $U_{f}U_{f}^{*} = I_{N_{f}}$ does not 
necessarily possess a smooth manifold structure. This complicates the application
of standard methods for the minimization of functions over
orthogonality constraints. 

On the other hand, the condition $4$ in
Proposition~\ref{prop:projection} represents $U$ in a factorized form,
hereinafter referred to as the $(X,Y)$ representation.
This representation of the matrix $U(\vk)$ gives rise to Wannier
functions composed of the $N_{f}(\vk)$ functions in the
frozen window, and another set of $N_{w} - N_{f}(\vk)$
functions, encoded by the matrix $Y(\vk)$. This $Y$ encapsulates all
the necessary information about the projector $P_{w} = \wt \Psi \wt \Psi^{*}$. The unitary
$X(\vk)$ matrix mixes these $N_{w}$ Wannier functions amongst themselves to
produce a smooth gauge. In the $(X,Y)$ representation,
the variational formulation~\eqref{eqn:varopt} can be written
as
\begin{equation}
  \begin{split}
  \inf_{\{X(\vk),Y(\vk)\}} &\quad \Omega[\{U(\vk)\}],\\
  \text{s.t.} & \quad U(\vk) = \begin{bmatrix} I_{N_{f}(\vk)} & 0 \\
    0 & Y(\vk)\end{bmatrix} X(\vk), \\
  & \quad X^{*}(\vk) X(\vk) = I_{N_{w}},\\ &\quad Y^{*}(\vk) Y(\vk) =
  I_{N_{w}-N_{f}(\vk)}.
  \end{split}
  \label{eqn:varpartial}
\end{equation}
This optimization problem is equivalent to the ``partly occupied Wannier
functions''~\cite{ThygesenHansenJacobsen2005}. This also directly
generalizes the maximally localized Wannier functions
procedure~\cite{MarzariVanderbilt1997} by Marzari and Vanderbilt for the
isolated case.

The $(X,Y)$ representation is a redundant representation, and a given $U$
can be reproduced by many pairs $(X,Y)$. However,
in contrast to the formulation~\eqref{eqn:varopt}, the constraint in
Eq.~\eqref{eqn:varpartial} defines a Riemannian manifold where
$X(\vk)$ and $Y(\vk)$ are independent matrices with orthogonality
constraints. This allows us to use standard algorithms for the
minimization of differentiable functions on Riemannian manifolds to
solve the problem. We refer to Appendix~\ref{app:gradient} for the details
of the computation of the gradient of the objective function $\Omega$.

\subsection{Implementation}

For our implementation, we modified the Julia \cite{bezanson2017julia}
library \texttt{Optim.jl} for unconstrained optimization to
accommodate constraints represented by Riemannian manifolds
\cite{edelman1998geometry,absil2009optimization}. Our modifications
have been integrated into that library and are available online
\footnote{\url{https://github.com/JuliaNLSolvers/Optim.jl}}. For the
numerical tests that follow we used the limited-memory BFGS algorithm
\cite{nocedal2006numerical} with Hager-Zhang line search
\cite{hager2005new}, which gave the best performance compared to other
readily-available algorithms (steepest descent, conjugate gradient,
BFGS) and line searches (fixed step, backtracking).

In order to generate the initial guess for numerical optimization, we
need to convert a given matrix $U$ to a pair $(X,Y)$ that
parametrizes it. It will also be useful to consider matrices $U$ that
do not satisfy the constraints $U^{*}U=I_{N_{w}}$ and $U_{f}U_{f}^{*}=I_{N_{f}}$ exactly but only approximately. 
This will allow us to project $U$ to the admissible set that satisfy
these constraints.

To find a pair $(X,Y)$ that represents a given $U$, we
first choose $Y$ to minimize the error on the projector $UU^{*}$
measured by the Frobenius norm via
\begin{equation}
  \inf_{Y^{*}Y=I} \left\Vert UU^{*}-\begin{bmatrix}
    I_{N_{f}} & 0 \\
    0 & Y Y^{*} 
  \end{bmatrix}\right\Vert_{F}^2.
  \label{}
\end{equation}
A solution to this problem can be computed using the eigenvalue decomposition
\begin{equation}
  U_{r}U_{r}^{*} = V S V^{*}.
  \label{}
\end{equation}
When $U$ satisfies the constraints, $U_{r}U_{r}^{*}$ is a
projector of rank $N_{w}-N_{f}$, and we can choose $Y$ to be the
columns of $V$ corresponding to the $N_{w}-N_{f}$ non-zero eigenvalues
of $S$. When $U$ does not satisfy the constraints, we pick $Y$ as the
columns of $V$ corresponding to the largest $N_{w}-N_{f}$ eigenvalues
of $S$.

Once $Y$ is computed, we can find $X$ that minimizes the error on $U$:
\begin{equation}
  \inf_{X^{*}X=I} \left\Vert U-\begin{bmatrix}
    I_{N_{f}} & 0 \\
    0 & Y 
  \end{bmatrix} X\right\Vert_{F}^2.
  \label{}
\end{equation}
When $U$ satisfies the constraints, the solution of this problem
is simply
\[
X
=  \begin{bmatrix}
    I_{N_{f}} & 0 \\
    0 & Y^{*}
  \end{bmatrix} U.
  \] 
Otherwise, the solution can again be obtained via the singular value decomposition
\begin{equation}
 \begin{bmatrix}
    I_{N_{f}} & 0 \\
    0 & Y^{*}
  \end{bmatrix} U = \wt{V_{\rm l}} \wt{S} \wt{V_{\rm r}}^{*},
  \label{}
\end{equation}
and setting $X = \wt{V_{\rm l}} \wt{V_{\rm r}}^*$. This step is also
called the L{\"o}wdin orthogonalization procedure~\cite{Loewdin1950}.

\subsection{Relation to disentanglement}
Our variational formulation also gives us a concise way to understand the
``disentanglement'' procedure of Souza, Marzari and
Vanderbilt~\cite{SouzaMarzariVanderbilt2001}, in which the spread
functional is split into two parts
\begin{equation}
  \Omega[\{U(\vk)\}] = \Omega_{I}[\{U(\vk)\}] +
  \widetilde{\Omega}[\{U(\vk)\}].
  \label{}
\end{equation}
Here $\Omega_{I}$ is called the gauge-invariant part (depending on $P_{w}$,
and hence only on $Y Y^{*}$), and $\widetilde{\Omega}$ is called the gauge
dependent part (depending on $X$).
Instead of optimizing~\eqref{eqn:varpartial}
directly,~\cite{SouzaMarzariVanderbilt2001} proposes to use a two step
procedure. First one optimizes the gauge-invariant part only:
\begin{equation}
  \inf_{\{Y(\vk)\}} \quad \Omega_{I}[\{U(\vk)\}].
  \label{}
\end{equation}
This is numerically expedient as $\Omega_{I}$ only depends on
$Y Y^{*}$. In fact, it is analogous to minimization problems in electronic
structure (for instance the Hartree-Fock model), where one minimizes
the energy, which only depends on the spectral subspace, over all
possible orthogonal orbitals. The authors in
\cite{SouzaMarzariVanderbilt2001} accordingly obtain a nonlinear eigenvalue
problem as the first-order optimality conditions, which they solve
using a damped self-consistent field (SCF) iteration.

Once $\{Y(\vk)\}$ is obtained, it is fixed and so is the projector
$P_{w}$. A second minimization problem
\begin{equation}
  \inf_{\{X(\vk)\}} \quad \widetilde{\Omega}[\{U(\vk)\}] 
  \label{}
\end{equation}
is then solved with respect to the gauge matrix $X(\vk)$. This
optimization problem is of
the same nature as the one for an isolated set of bands.

The total spread from the above two-step
procedure is necessarily larger or equal to the global minimum
of~\eqref{eqn:varpartial}. Interestingly, although the optimal spread can be
substantially lower than that found by the two-step disentanglement
procedure, numerical experiments show that the quality of Wannier
interpolation, measured for instance by the qualify of band structure
interpolation, is often similar in both cases.

\subsection{Selected column of the density matrix}
\label{sec:scdm}
While the primary purpose of this work is to introduce and analyze a
variational formulation of Wannier functions, both the objective
function and the constraints are nonlinear, and hence there may exist
multiple local minima.  It is practically important to seed
such methods with a good initial guess. Here, we summarize the
recently developed unified methodology for Wannier localization of
entangled band structure \cite{DamleLin2017} based on the selected
columns of the density matrix (SCDM) methodology \cite{DamleLinYing2015}. Importantly, this method is direct and robust\textemdash no initial guess is required and it will generate valid output\textemdash and thus may be reliably used to generate an initial guess.

The SCDM method for entangled band structure first constructs a
quasi-density matrix
\begin{equation}
\label{eqn:quasidensity}
f(H(\vk)) = \sum_{n=1}^{N_{o}}f\left(\varepsilon_{n,\vk}\right)
  \ket{\psi_{n,\vk}}\bra{\psi_{n,\vk}},
\end{equation}  
For insulating systems $f$ would be $1$ on the occupied bands and
$0$
otherwise, yielding the projector $P$ as before. For entangled band
structure however, the function $f(\cdot)$ is chosen to be large on the bands of
interest and decays rapidly, but smoothly, away from them
\cite{DamleLin2017}. The SCDM algorithm constructs a gauge
by selecting a common set of columns of the $\vk$-dependent
(quasi-)density matrix $f(H(\vk))$. In practice, it is often sufficient to
select these columns based on an ``anchor'' point denoted $\vk_0$\textemdash
generically chosen to be the so-called Gamma-point $(0,0,0)^T.$

We now briefly outline the SCDM method and refer the reader to
\cite{DamleLin2017} for more details. Let
$\Psi_{\vk}\in\mathbb{C}^{N_g\times N_o}$ be the matrix with
orthogonal columns that represents $\left\{\psi_{n,\vk}(\vr)\right\}$ on
a discrete grid in the unit cell, and $\mc{E}(\vk) =
\text{diag}\left[\left\{\varepsilon_{n,\vk}\right\}_{n=1}^{N_o}\right]$
be a diagonal matrix encoding the corresponding eigenvalues.
SCDM identifies $N_w$ columns of $f(H(\vk))$ based on the leading $N_w$
columns of the permutation matrix $\Pi$, computed via the
QR factorization with column pivoting (QRCP) procedure
\begin{equation}
\Psi_{\vk_0}^*\Pi = QR.
\end{equation}
This set of columns is denoted by $\mc{C} =
\left\{\vr_{n}\right\}_{n=1}^{N_w}\subset \Gamma$. Now, for each $\vk$-point define $\Xi(\vk)\in\mathbb{C}^{N_o\times N_w}$ as
\begin{equation}
\Xi_{n,n'}(\vk) = f(\varepsilon_{n,\vk})\psi_{n,\vk}^*(\vr_{n'}).
\end{equation} 
It is expected that 
\begin{equation}
\wt{\psi}_{n,\vk}(\vr) = \sum_{m}\psi_{m,\vk}(\vr)\Xi_{m,n}(\vk)
\end{equation}
is smooth with respect to $\vk.$ Therefore, if the singular values of
$\Xi(\vk)$ are uniformly bounded away from 0 in the Brillouin zone,
$U(\vk)$ constructed via L{\"o}wdin orthogonalization~\cite{Loewdin1950} of $\Xi(\vk)$ has
orthogonal columns, and yields $\left\{\tilde{\psi}_{n,\vk}\right\}$
that are smooth with respect to $\vk$. 

In this framework the frozen bands are
not represented exactly, so prior to use in our optimization procedure we must convert from $U$ to a pair $(X,Y)$ via our aforementioned scheme. Since the projector on the
frozen set varies discontinuously with $\vk$, this procedure does not
produce a continuous gauge. However, if the function $f$ is chosen appropriately, it should be close to one. This is further corroborated by the quality of band interpolation, despite the substantially larger spread, we achieve in the numerical results section using the SCDM initial guess without explicitly freezing any bands.


\section{Real materials}\label{sec:materials}
We first consider the performance of our variational formulation for
several real materials. This includes valence and conduction bands of
silicon (semiconductor), conduction bands of copper (metal), and valence
bands of aluminum (metal). We always start with the aforementioned SCDM
based initial guess.  We compare the result obtained from the
variational formulation to that obtained from the disentanglement
procedure in \wannierp, as well as the result obtained directly from the
SCDM initial guess without further refinement. 

In these experiments, the choice of the
parameters of the SCDM procedure are chosen to yield good baseline band
structure interpolation. However, they are not ``optimized'' to minimize band structure interpolation error. These experiments often show how the two optimization methods are
comparable, though in some situations we are able to find Wannier functions with
smaller spreads using our variational method even given the same initial guess. One interesting
point that we will see play out throughout our examples is that the
value of the spread and band interpolation quality may not be directly
connected, \emph{i.e.}\ 
Wannier functions with considerably different values of spread can yield
qualitatively similar interpolation error.

All of the $\vk$-point calculations and reference band structure
calculations were performed with
\qe~\cite{GiannozziBaroniBoniniEtAl2009}. The SCDM initial guess was
constructed using the code available
online\footnote{\url{https://github.com/asdamle/SCDM}}. Our new variational
formulation was implemented in the Julia language and is available
online\footnote{\url{https://github.com/antoine-levitt/wannier}}.

\subsection{Silicon}
Here, we compute the lowest 16 bands of Silicon on an $8\times 8\times 8$ $\vk$-point grid and then proceed to
compute eight Wannier functions. This includes the four valence bands
and four additional low lying conduction bands. For the SCDM procedure
we use $\mu = 11.0$ and $\sigma = 2.0$ with $f$ corresponding to
``entangled case 1'' in \cite{DamleLin2017}\textemdash a complementary error function. For \wannier and our
method,
we freeze bands below 12 eV and set the outer window
maximum at infinity. For \wannier the prescribed convergence criteria
of $1\times 10^{-10}$ on the spread was reached after 225
disentanglement iterations and 95 spread reduction iterations, and for
our method after 149 iterations.

Figure~\ref{fig:Si8_band_interp} shows the band structure interpolation
using the three methods. We see that for the four
valence bands all three methods perform very well. Furthermore, while
there are differences in the interpolation of the conduction bands, no
one method clearly outperforms the others. As expected, if we consider
the total spread (see Table~\ref{tab:Si_spread}) of the final localized
orbitals, our variational formulation yields the most compact orbitals.

In Table~\ref{tab:Si_spread_all} we report the per orbital spread for
each method, and observe that to the number of significant digits
reported all the orbitals found by our method have the same spread (they
do not vary until the fourth decimal place). In contrast, \wannier seems
to converge to two distinct sets of orbitals with slightly different
spreads. In Figure~\ref{fig:Si_orbitals}, we illustrate the differences
between the orbitals found with our variational method and those found
via \wannierp.  We observe that the orbitals obtained from our
variational formulation resemble more closely to sp$^3$ hybridized
orbitals centering around each Si atom, as indicated from 
chemical intuition.  Interestingly, by using the output of our variational
method as input to \wannier we are able to force \wannier to converge to
the same point as our method. Unfortunately, it is difficult to pinpoint
a specific cause for the apparent convergence of \wannier to a worse
local minimum in this setting.

\begin{figure}[h!]
  \centering
  \includegraphics[width=.8\textwidth]{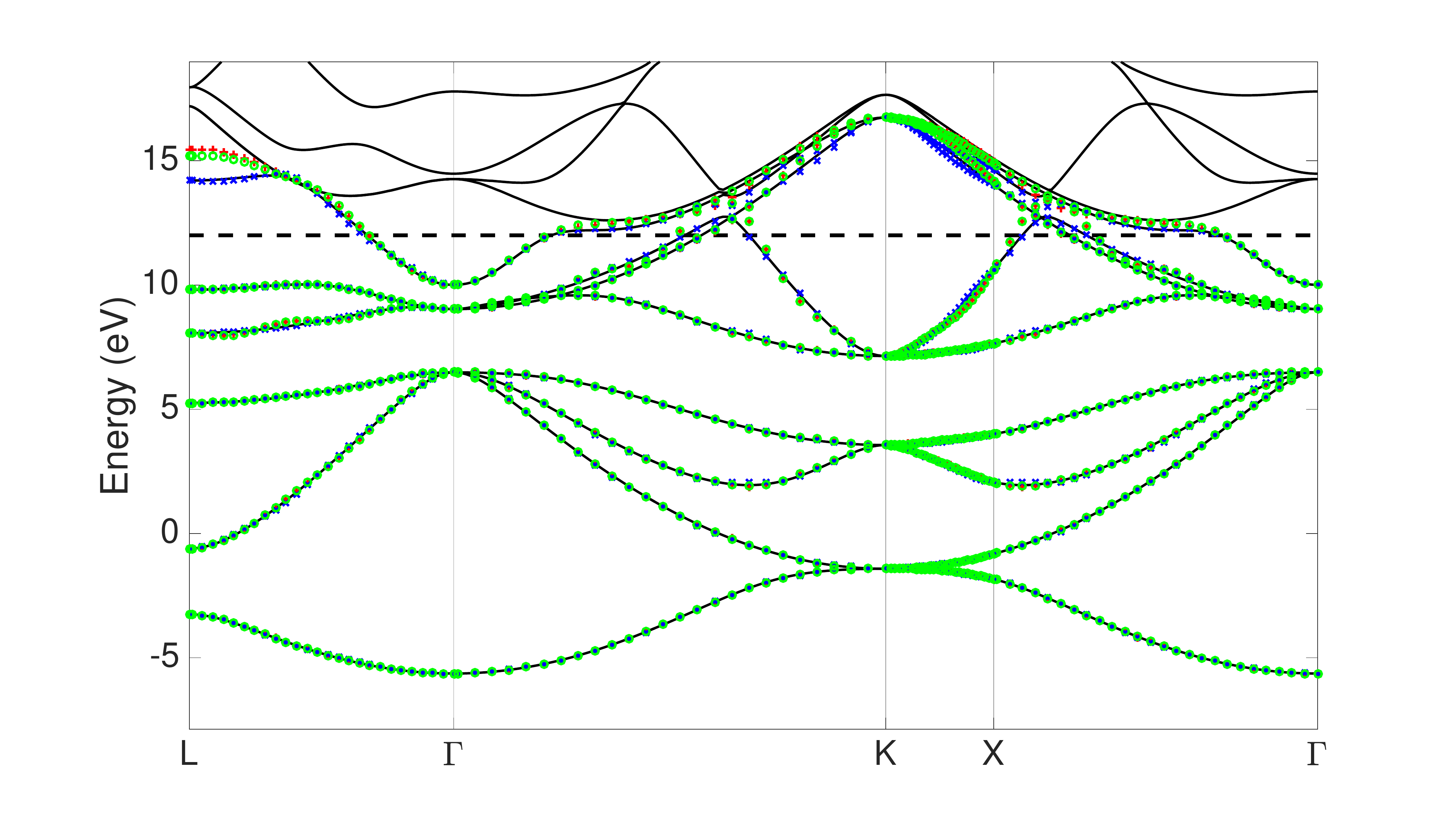}
  \caption{Wannier interpolation of Silicon with 8 k-points using (blue Xs) SCDM, (green circles) our variational formulation, and (red +s) \wannier compared with a (black line) reference calculation. The frozen window is the region below the dotted black line.}
  \label{fig:Si8_band_interp}
\end{figure}

\begin{table}
  \caption{Silicon spread and valence band error comparison}
\label{tab:Si_spread}
\rowcolors{2}{gray!25}{white}
\centering
\begin{tabular}{lccc} \rowcolor{gray!50}
   & Final spread $\left(\text{\AA}^2\right)$ & max error (eV) & RMSE (eV) \\
 Variational & 25.177 & 0.069 & 0.021 \\ 
 \wannier & 27.00 & 0.083 & 0.023 \\ 
 SCDM & 45.206 & 0.112 & 0.029 \\ 
\end{tabular}
\end{table}

\begin{table}
\caption{Spreads of the eight individual Wannier functions for silicon}
\label{tab:Si_spread_all}
\rowcolors{2}{gray!25}{white}
\centering
\begin{tabular}{lcccccccc} \rowcolor{gray!50}
 & \multicolumn{8}{c}{Orbital spread $\left(\text{\AA}^2\right)$} \\
 Variational & 3.15 & 3.15 & 3.15 & 3.15 & 3.15 & 3.15 & 3.15 & 3.15 \\ 
 \wannier & 3.16 & 3.16 & 3.16 & 3.16 & 3.59 & 3.59 & 3.59 & 3.59 \\ 
 SCDM & 4.93 & 4.93 & 4.93 & 4.93 & 6.37 & 6.37 & 6.37 & 6.37 \\ 
\end{tabular}
\end{table}

\begin{figure}[h!]
  \centering
  \includegraphics[width=0.45\textwidth]{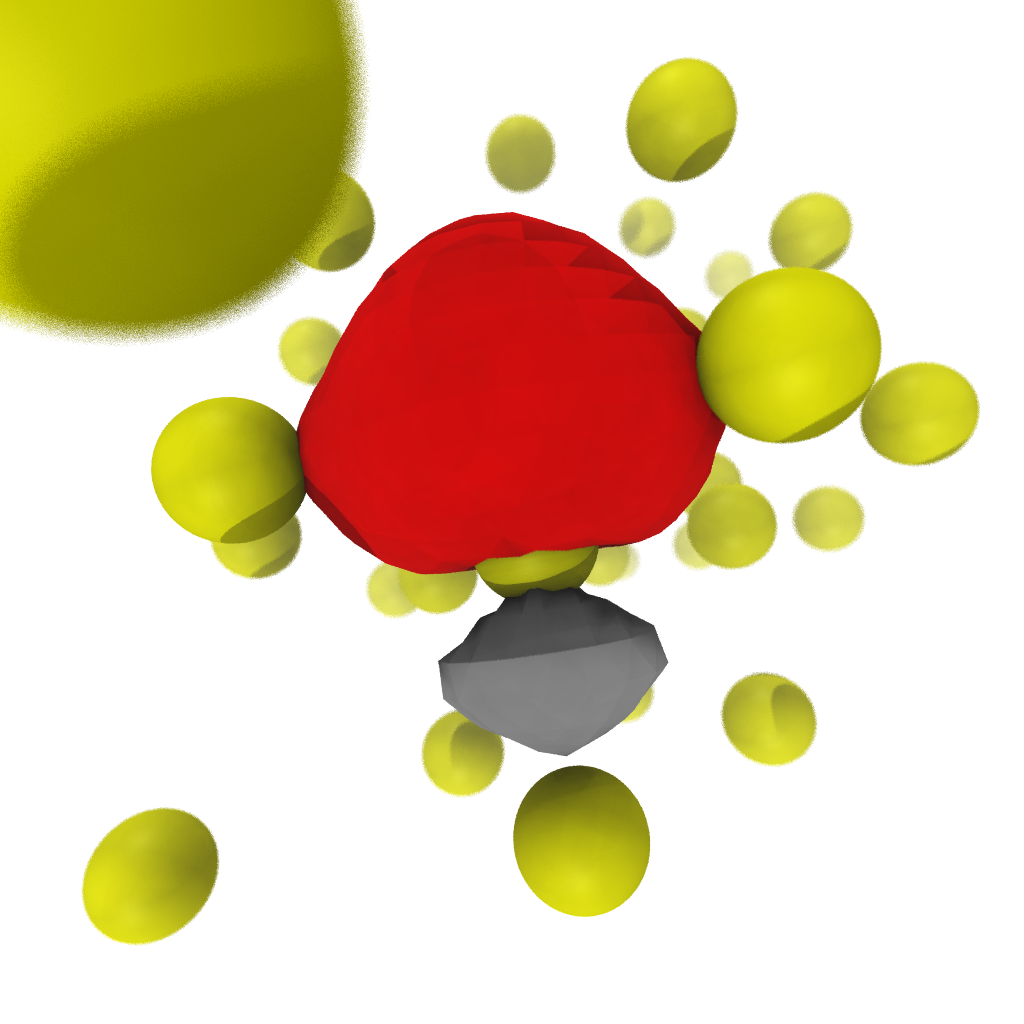}
  \includegraphics[width=0.45\textwidth]{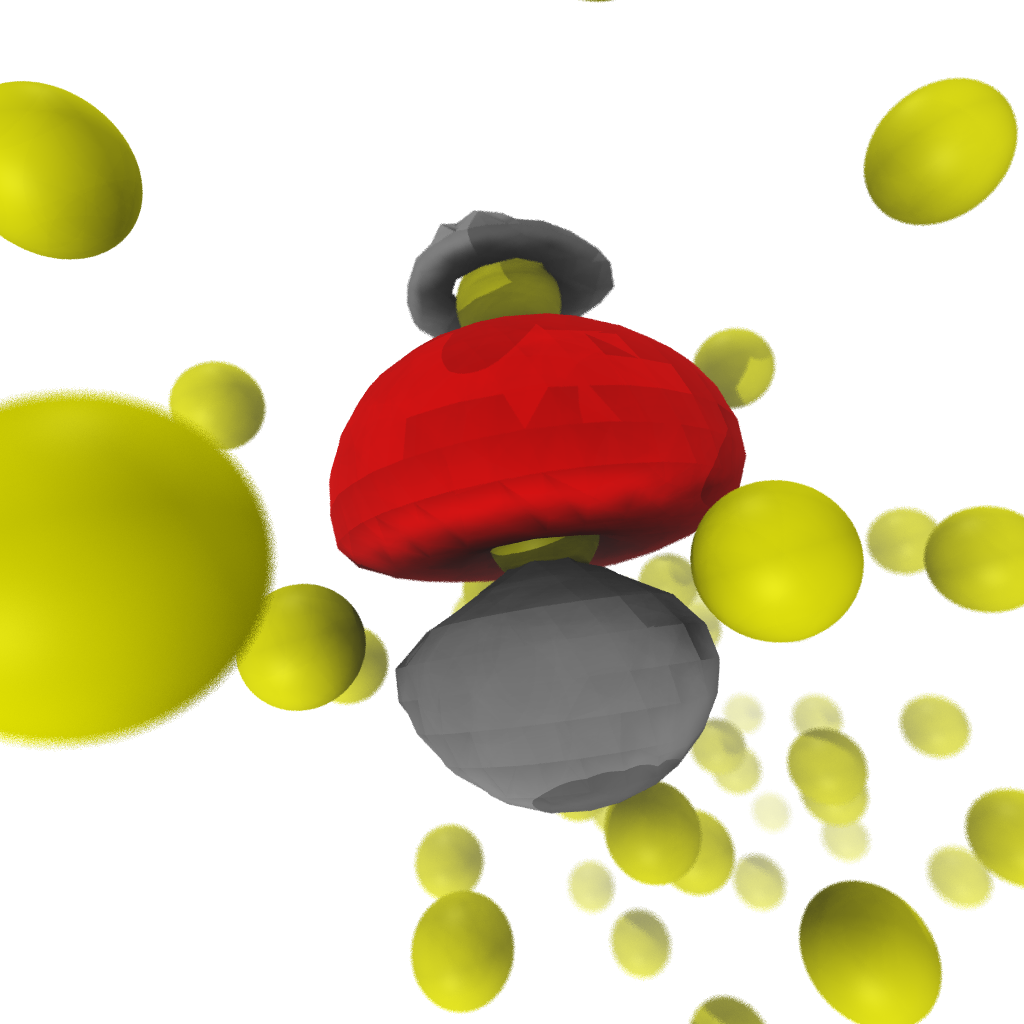}
  \caption{Example orbitals from our variational method (left) and \wannier (right). Here the red and gray isosurfaces are plotted at values $\pm 0.5$ for both normalized orbitals. All of the orbitals we find with our variational method seem to clearly have sp$^3$ hybrid character, as seen in the left figure. In contrast, as illustrated on the right, some of the orbitals found by \wannier with larger spread do not share this behavior as clearly.}
  \label{fig:Si_orbitals}
\end{figure}



\subsection{Aluminum}
We now repeat the same experiments as before, but for the valence bands
of the aluminum system with an $8\times 8\times 8$ $\vk$-point grid. Specifically, we start with six bands and seek
four Wannier functions. For the SCDM procedure we use $\mu = 8.42$ and
$\sigma = 4.0$ with $f$ corresponding to ``entangled case 1'' in
\cite{DamleLin2017}\textemdash a complementary error function. Once again, \wannier was run until convergence at $10^{-10}$ for both the disentanglement and spread minimization, and reached that threshold after 2,437 and 91 iterations. Our method converged to a spread reduction tolerance of $10^{-10}$ after 138 iterations. As we observe in Figure~\ref{fig:Al_band_interp}, all three methods once again perform well\textemdash particularly below
the Fermi energy. The final spreads for each of the three methods are
reported in Table~\ref{tab:Al_spread} along with the spreads of each orbital. While both \wannier and our
method improve upon the spread of the SCDM initial guess, we do find
slightly smaller spread with our optimization procedure. Here, bands
below 11.6 eV were frozen in both \wannier and our method, and the outer
window was set to $\infty.$

\begin{figure}[h!]
  \centering
  \includegraphics[width=.8\textwidth]{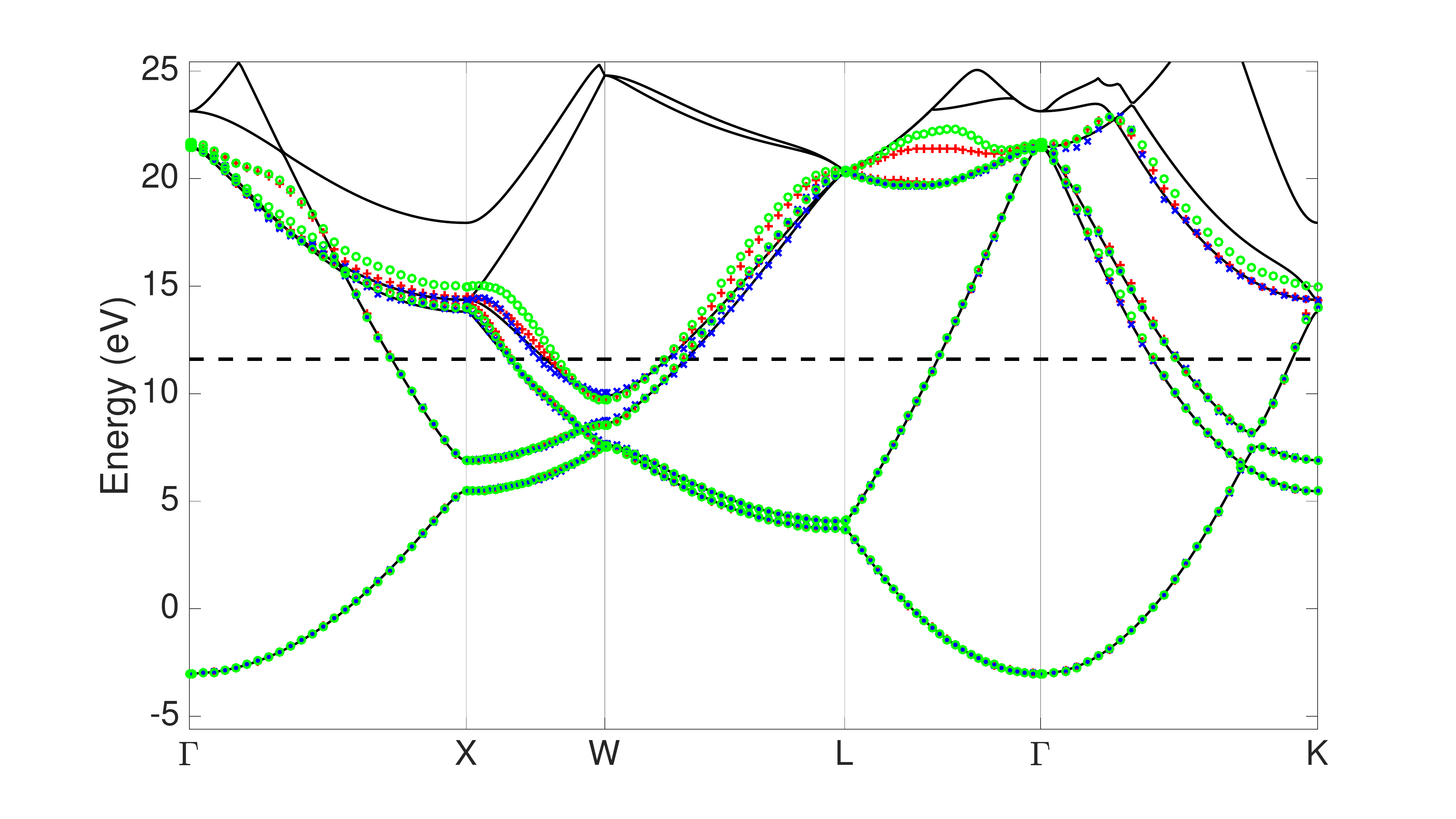}
  \caption{Wannier interpolation of Aluminum with 8 $\vk$-points per direction using (blue Xs) SCDM, (green circles) our variational formulation, and (red +s) \wannier compared with a (black line) reference calculation. The frozen window is the region below the dotted black line.}
  \label{fig:Al_band_interp}
\end{figure}


\begin{table}
\rowcolors{2}{gray!25}{white}
\centering
\begin{tabular}{lccccc} \rowcolor{gray!50}
 & \multicolumn{4}{c}{Orbital spread $\left(\text{\AA}^2\right)$} & Final spread $\left(\text{\AA}^2\right)$\\
 Variational & 2.01 & 2.02 & 2.02 & 2.03 & 8.07\\ 
 \wannier & 2.09 & 2.1 & 2.11 & 2.11 & 8.41\\ 
 SCDM & 3.44 & 4.02 & 4.21 & 4.21 & 15.89\\ 
\end{tabular}
\caption{Spreads of the four individual Wannier functions for Aluminum and the final spread}
\label{tab:Al_spread}
\end{table}

We also consider the convergence of the band interpolation error in this setting, looking at both maximum error and RMSE. As before, for both our method and \wannier we freeze bands below 11.6 eV and set the outer window to infinity. For all three methods we then measure the error of band interpolation at or below the Fermi energy (8.42 eV). In all cases, we capped \wannier at 5,000 disentanglement and spread reduction iterations and considered it converged at a tolerance of $10^{-10}.$ Similarly, we considered our method converged the objective function changed by less than $10^{-10}$ between successive iterates, or we reached 5,000 iterations. Typically our method took 100-250 iterations to converge, whereas the \wannier disentanglement would hit the iteration cap and spread reduction took 90-400 iterations to converge. Figure~\ref{fig:Al_converge} shows broadly similar behavior for all three methods, though generally the two optimization methods do noticeably improve on the SCDM initial guess as more $\vk$-points are used. We expect that asymptotically the optimization based methods should perform better. However, given the relatively small number of grid points per direction and the complexity of the band structure, it seems we are still in the pre-asymptotic regime.

\begin{figure}[h!]
  \centering
  \includegraphics[width=0.49\textwidth]{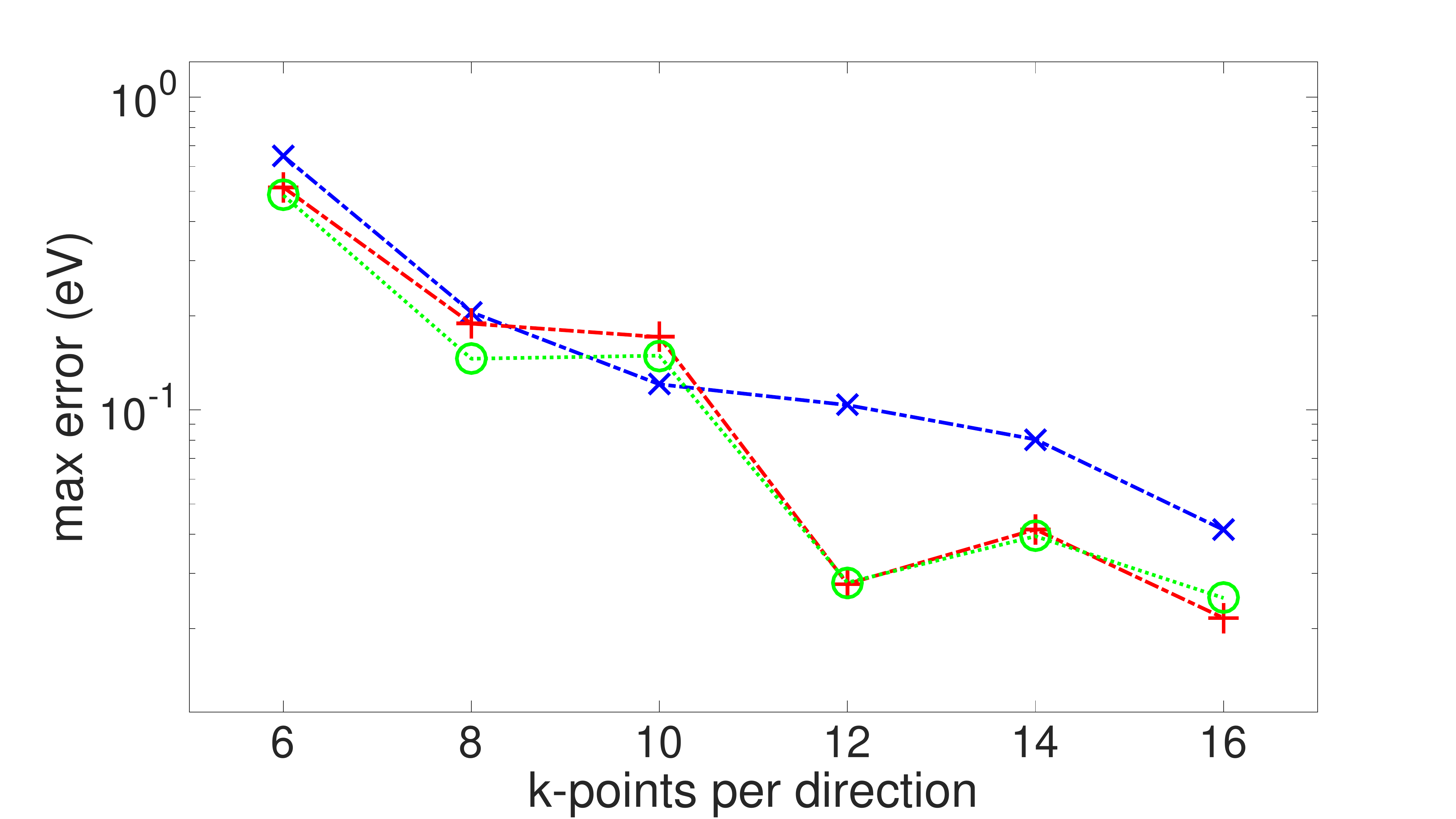}
  \includegraphics[width=0.49\textwidth]{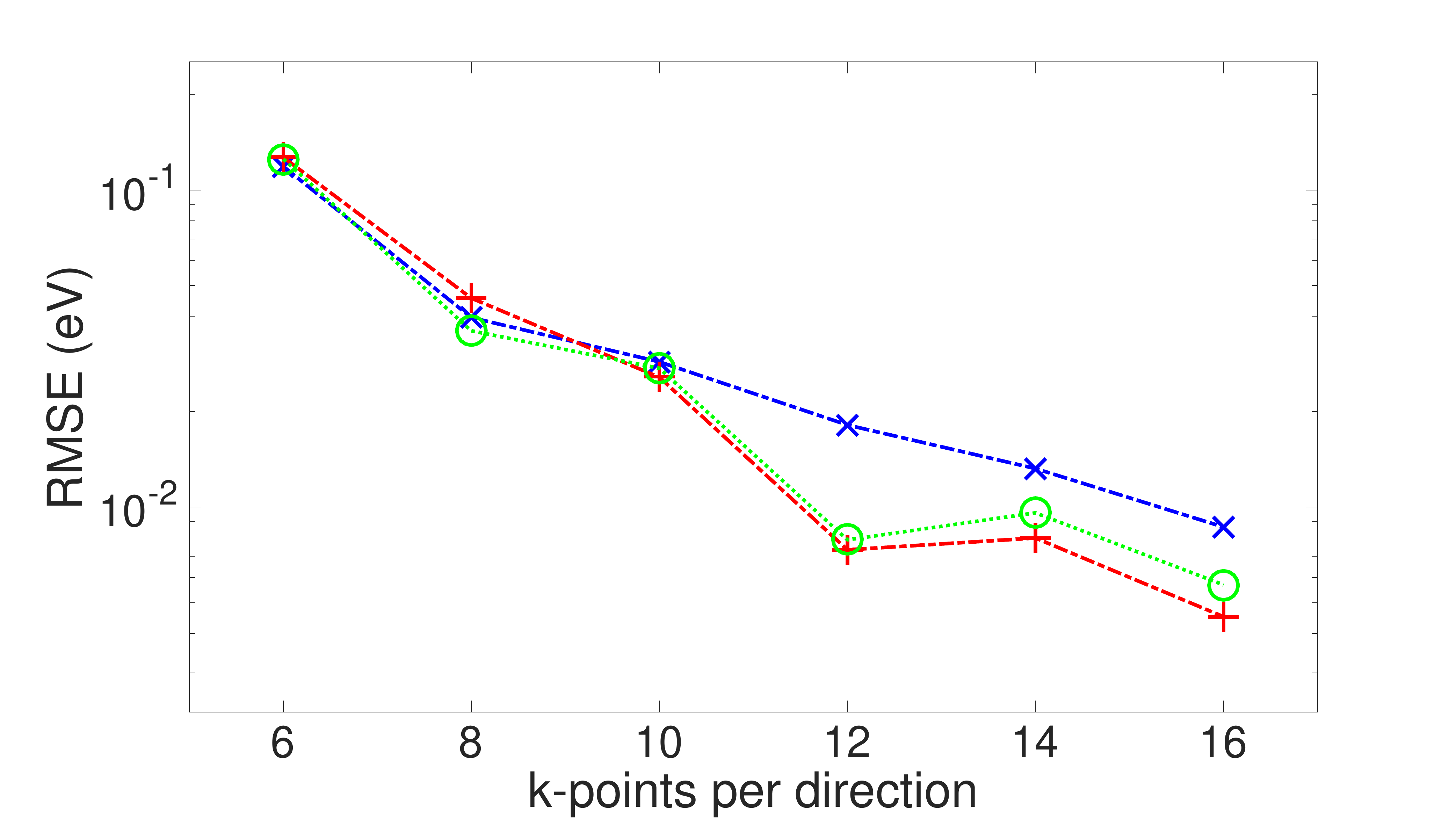}
  \caption{Max and RMSE in band interpolation for aluminum
  as computed using (blue Xs) SCDM, (green circles) our variational
  formulation, and (red +s) \texttt{Wannier90}.}
  \label{fig:Al_converge}
\end{figure}

\subsection{Copper}
Lastly, we illustrate the behavior of our method when used to
interpolate seven conduction bands of copper around the Fermi energy.
Because several bands that we do not wish to track cross through the
energy window, it is not easy to simply look at band interpolation
error. Rather, we place emphasis on the qualitative behavior of the
interpolation and the orbital spread yielded by our variational
formulation. Interestingly, in this case we observed a high
sensitivity of \wannier to the SCDM initial guess based on the
parameters (corresponding to ``entangled case 2'' in
\cite{DamleLin2017}\textemdash a Gaussian) and a contrasting
robustness of our variational method. In all cases we used a frozen
window of 13.5 to 17 eV and no outer window for both \wannier and our variational method.

When fixing the parameter $\mu = 15.5$ and varying $\sigma$ from $3.0$
to $6.0$, both SCDM and our variational method robustly generated good
band interpolation. However, for a range of $\sigma$ \wannier with
disentanglement failed to reach convergence. To sweep over several
values of $\sigma$ we limit \wannier to 5,000 iterations each for the
disentanglement and spread minimization, and we limit our variational
method to 1,000 iterations. 

\begin{myrem}
For $\sigma = 5.0$ where we
observed particularly bad performance of \wannier (see below), we let it run for
100,000 iterations of each the disentanglement and spread. While the
disentanglement procedure converged after roughly 20,000 iterations,
the spread minimization failed to converge even after 100,000
iterations. While this does not guarantee that the local minimum
\wannier may eventually find is poor and could simply be the
optimization algorithm behaving poorly, we feel this is a reasonable
comparison to make even without \wannier determining that it has
converged. Interestingly, passing the output of \wannier in this setting to our variational method we were able to converge to a good gauge.
\end{myrem}

Figures~\ref{fig:Cu_band_interp40} and~\ref{fig:Cu_band_interp50} show
the band interpolation of the three methods in the case where $\sigma =
4.0$ and $\sigma = 5.0$. We also report the individual spreads of the orbitals in Tables~\ref{tab:Cu_spread_all4} and~\ref{tab:Cu_spread_all5}.  We observe that in both cases the SCDM initial
guess and our optimized solution yield good band interpolation within
the frozen window. In
contrast, \wannier does not find a good local optima in the latter case
and this results in poor interpolation quality. We further investigate
this behavior in Table~\ref{tab:Cu_spread}, where we report the total
final spread of the methods as we vary $\sigma$. We see that for two of
the parameter values \wannier failed to find a local optimum that is close
to what our variational method finds.


\begin{figure}[h!]
  \centering
  \includegraphics[width=0.8\textwidth]{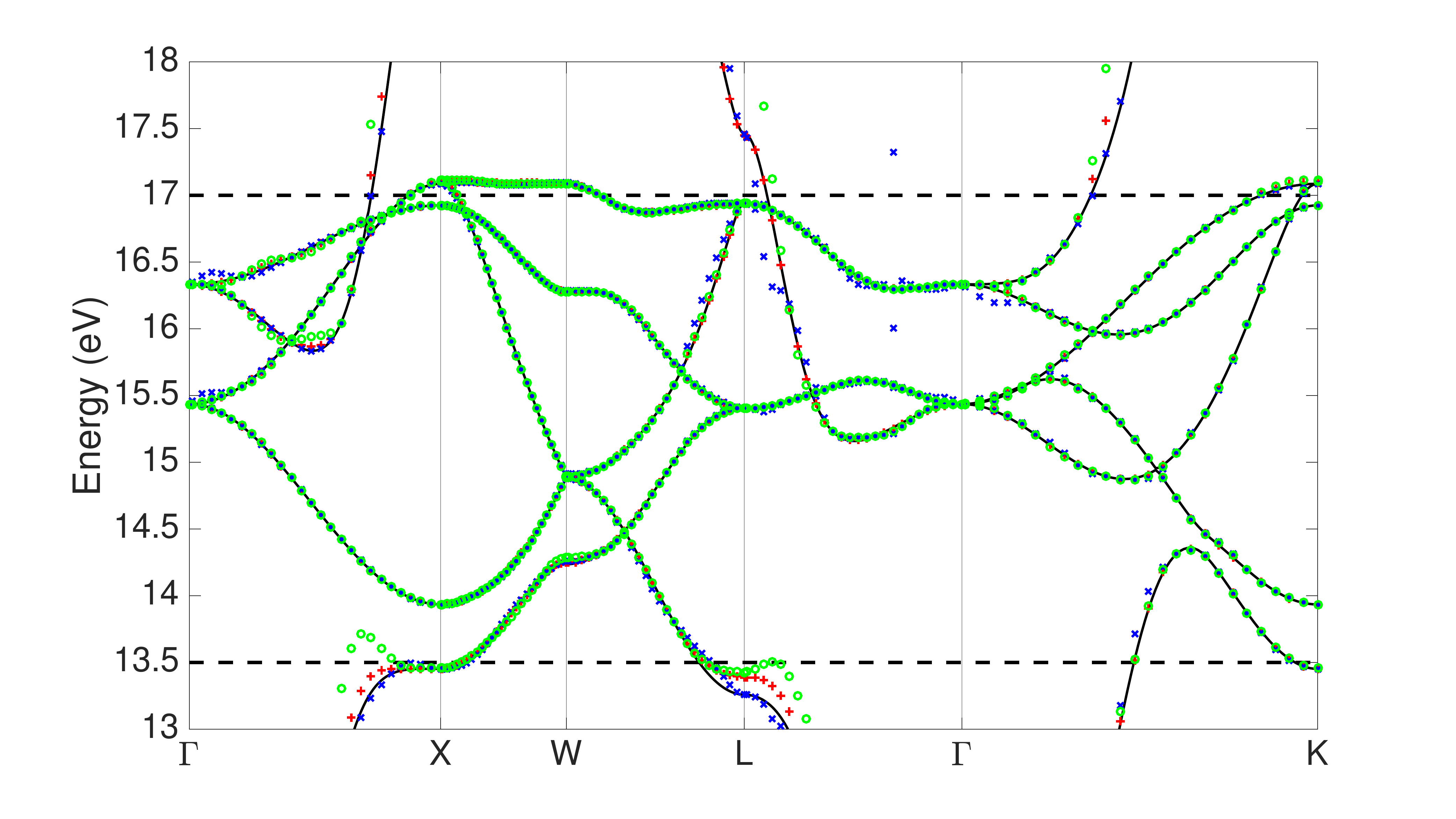}
  \caption{Using an SCDM initial guess with $\mu = 15.5$ and $\sigma =4.0$, Wannier interpolation of copper with 10 k-points using (blue Xs) SCDM, (green circles) our variational formulation, and (red +s) \wannier compared with a (black line) reference calculation. The frozen window is the region between the dotted black line.}
  \label{fig:Cu_band_interp40}
\end{figure}

\begin{table}
\rowcolors{2}{gray!25}{white}
\centering
\begin{tabular}{lccccccc} \rowcolor{gray!50}
 & \multicolumn{7}{c}{Orbital spread $\left(\text{\AA}^2\right)$} \\
 Variational & 0.24 & 0.50 & 0.50 & 0.51 & 0.56 & 0.56 & 1.30 \\ 
 \wannier & 0.41 & 0.43 & 0.43 & 0.48 & 0.54 & 0.56 & 1.41 \\ 
 SCDM & 1.49 & 2.34 & 2.46 & 3.06 & 3.13 & 3.68 & 7.50 \\ 
\end{tabular}
\caption{Spreads of the seven individual Wannier functions for copper with the $\sigma = 4.0$ initial guess.}
\label{tab:Cu_spread_all4}
\end{table}

\begin{table}
\rowcolors{2}{gray!25}{white}
\centering
\begin{tabular}{lccccccc} \rowcolor{gray!50}
 & \multicolumn{7}{c}{Orbital spread $\left(\text{\AA}^2\right)$} \\
 Variational & 0.21 & 0.47 & 0.51 & 0.52 & 0.53 & 0.62 & 1.38 \\ 
 \wannier & 0.47 & 0.53 & 0.54 & 0.60 & 0.71 & 2.78 & 2.83 \\ 
 SCDM & 1.21 & 1.47 & 1.56 & 1.78 & 1.97 & 2.53 & 8.87 \\ 
\end{tabular}
\caption{Spreads of the seven individual Wannier functions for copper with the $\sigma = 5.0$ initial guess.}
\label{tab:Cu_spread_all5}
\end{table}

\begin{figure}[h!]
  \centering
  \includegraphics[width=0.8\textwidth]{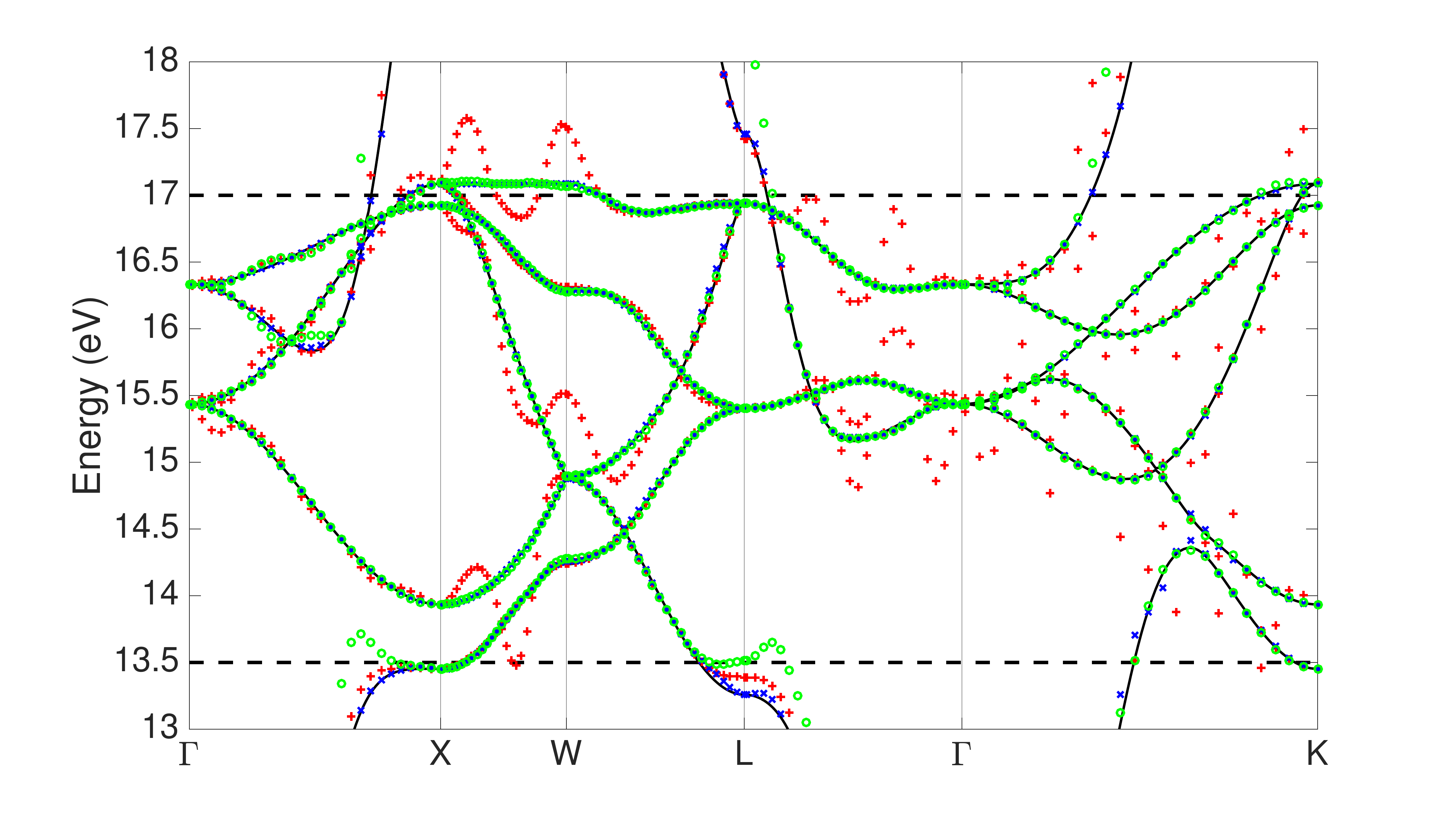}
  \caption{Using an SCDM initial guess with $\mu = 15.5$ and $\sigma =5.0$, Wannier interpolation of Copper with 10 k-points using (blue Xs) SCDM, (green circles) our variational formulation, and (red +s) \wannier compared with a (black line) reference calculation. The frozen window is the region between the dotted black line.}
  \label{fig:Cu_band_interp50}
\end{figure}

\begin{table}
\rowcolors{2}{gray!25}{white}
\centering
\begin{tabular}{lccccccc} \rowcolor{gray!50}
 $\sigma$  & 3.0 & 3.5 & 4.0 & 4.5 & 5.0 & 5.5 & 6.0  \\
 SCDM & 34.22 & 26.9 & 23.65 & 20.83 & 19.39 & 18.7 & 17.3 \\ 
 Our method & 4.23 & 4.17 & 4.18 & 4.23 & 4.23 & 4.18 & 4.18 \\ 
 \wannier & 4.27 & 4.26 & 4.26 & 8.48 & 8.46 & 4.26 & 4.26 \\ 
\end{tabular}
\caption{Comparison of the spreads $\left(\text{\AA}^2\right)$ for copper as the $\sigma$ parameter of the SCDM initial guess is varied.}
\label{tab:Cu_spread}
\end{table}

\section{Free electron gas}\label{sec:electrongas}
We now investigate the decay property of the generalized Wannier
functions in real space. This is hard to
investigate numerically for real materials because of the very large
number of $\vk$ points needed to see the asymptotic decay, as
evidenced in the previous section. For
instance, in \cite{YatesWangVanderbiltEtAl2007} the authors report a
fast (consistent with exponential) decay up to grid sizes of
$15 \times 15 \times 15$, although the convergence appears to slow
down after that. Another problem is that it is hard to find realistic
systems in one or two dimensions on which disentangled Wannier systems
make sense. The free electron gas, \ie~$V = 0$, is explicitly solvable
and poses a very interesting benchmark for disentanglement, even in
one or two dimensions. In this section, we therefore apply the
methodology of Section~\ref{sec:varwannier} to $d$-dimensional free
electron gas. We simply define the lattice as
$\mathbb L = 2\pi \mathbb Z^{d}$, so that
$\mathbb L^{*} = \mathbb Z^{d}$, and we let the Brillouin zone be the
set $[0,1)^{d}$.

In this case the eigenfunctions of the operator $H(\vk) = (-\I \nabla
+ \vk)^{2}$ are given by
\begin{align*}
  v_{\vK}(\vr) &= \frac{1}{\sqrt{|\Gamma|}} e^{\I \vK \cdot \vr}
\end{align*}
for $\vK$ in the reciprocal lattice $\mathbb Z^{d}$. The corresponding
eigenvalues are $\varepsilon_{\vK,\vk} = |\vK + \vk|^{2}$. Notice that
since $\varepsilon_{\vK,\vk}$ is the squared distance from $\vk$ to
$-\vK$, the dispersion relation is the set of squared distances of
$\vk$ to the points of the reciprocal lattice $\mathbb
Z^{d}$.

We order the eigenfunctions by ordering the eigenvalues in a
non-decreasing order. We let
$u_{n,\vk} = e_{\vK_{n}}$, where $\varepsilon_{\vK_{n},\vk}$ is the
$n$-th eigenvalue of $H(\vk)$, this choice being arbitrary in the
presence of degeneracies. The matrix elements
$M_{mn}^{\vk,\vb} = \langle u_{m,\vk}, u_{n,\vk+\vb}\rangle$ of
overlap between neighboring $\vk$ points used in
the optimization process \cite{MarzariVanderbilt1997} then assume a
particularly simple expression: $M_{mn}^{\vk,\vb} = 1$ if $u_{m\vk}$
and $u_{n,\vk+\vb}$ are associated with the same $\vK$, and $M_{mn}^{\vk,\vb} = 0$
otherwise. In particular, this matrix differs from the identity (it is a permutation matrix) near eigenvalue crossings, where
$\varepsilon_{\vK,\vk} = \varepsilon_{\vK',\vk}$ with $\vK' \neq \vK$.

The free electron gas also makes it particularly easy to compute the
Wannier functions through their Fourier transforms. For a Wannier
function given in $\vk$-space by
\begin{align*}
  \widetilde{\psi}_{n\vk}(\vr) = \sum_{n \in \N} e^{\I \vk\cdot \vr}
 v_{\vK_{m}}(\vr)U_{mn}(\vk),
\end{align*}
it holds that
\begin{align*}
  w_{n,\bvec{0}}(\vr) = \int_{[0,1]^{d}} \widetilde{\psi}_{n\vk}(\vr) \ud \vk = \frac{1}{(2\pi)^{d}}\int_{[0,1]^{d}}\sum_{m \in \N} U_{mn}(\vk)e^{\I (\vk+\vK_{m})\cdot \vr},
\end{align*}
from which it follows that $U_{mn}(\vk)$ is simply the Fourier transform
of $w_{n,\bvec{0}}$ at frequency $\bvec{\xi} = \vK_{m}+\vk$.

Note that the free electron gas possesses a large number of symmetries. As a
consequence, it has a number of properties that are not expected from
generic systems. For instance, eigenvalue crossings are numerous (of
co-dimension 1, \ie~points in 1D, lines in 2D and planes in 3D), while
they are expected from the von Neumann-Wigner theorem
\cite{von1993merkwurdige} to be rare. This theorem predicts that, in
the absence of particular symmetries, the crossing of eigenvalues of a
Hermitian matrix are a phenomenon of codimension 3. This is believed
to be true (although unproved in many cases) for ``generic''
Schr\"odinger operators \cite{kuchment2016overview}: for a generic
$V$, the band structure is only expected to show crossings in
codimension $3$ (isolated points in 3D, and no crossings in lower
dimensions). The free electron gas, possessing all the symmetries a
Schr\"odinger operator can have, may thus be considered a worst case
system for disentanglement.

\subsection{The 1D case}

\subsubsection{One frozen band, two Wannier functions}
The dispersion relation of the 1D free electron gas results in a
crossing between the first and second eigenvalues at half-integers
and between the second and third eigenvalues at integers (see Figure
\ref{fig:WF_1D}). Because of the crossing between the first and second
band, any single Wannier function representing the first band has a
maximal decay of $1/r$. We therefore consider the problem of finding
two Wannier functions representing the first band
($N_{w} = 2, N_{f} = 1$). We do not set an outer window but discretize
the band structure using $2L+1$ Fourier modes with $L=10$ (as we will
see, the Wannier functions we find are compactly supported in Fourier
space, and therefore do not depend on the choice of
$L \geq 2$).

For this simple system, initializing the gauge randomly allows our
optimization algorithm to converge robustly to the same Wannier
functions, up to a change of sign and a shift by a lattice vector.
It would therefore seem that the global minimizer of the spread is
unique (up to the invariant degrees of freedom described above) and
real. The SCDM
algorithm with the settings $\mu = 0$ and $\sigma = 2$ yields Wannier
functions that are visually indistinguishable from the optimized
Wannier functions, and have a very similar spread ($2.51$ for the SCDM
algorithm as compared to $2.44$ for the optimized Wannier functions).

\begin{figure}[h]
  \centering
  \begin{subfigure}[t]{.45\textwidth}
    \includegraphics[width=\textwidth]{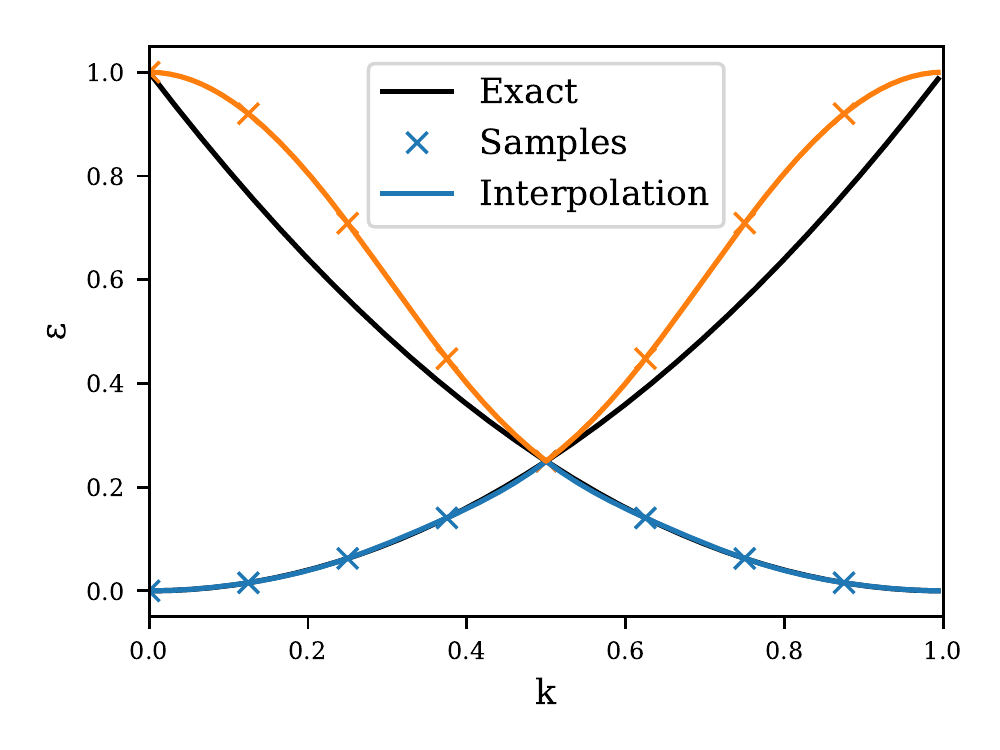}
    \caption{Exact band plot, samples on the $k$-space grid, and
      Wannier interpolation using the optimized Wannier
      functions.}
  \end{subfigure}
  \quad
  \begin{subfigure}[t]{.45\textwidth}
    \includegraphics[width=\textwidth]{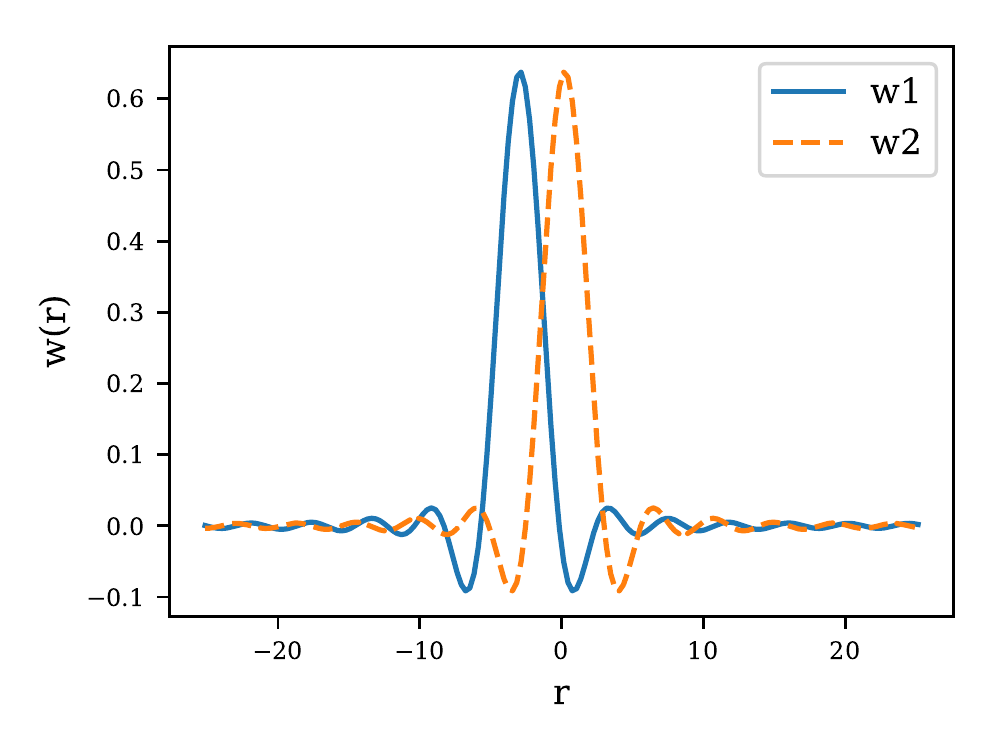}
    \caption{Wannier functions in real space.}
  \end{subfigure}
  \caption{Optimized Wannier functions obtained on a $k$-space grid of size $N=8$.}
  \label{fig:WF_1D}
\end{figure}

In Figure \ref{fig:WF_1D} we observe that the first band (which is
frozen) is exactly reproduced. Since the Wannier functions are
localized, the Wannier interpolation is very good, and in particular
has no trouble reproducing the crossing. Furthermore, the optimized
Wannier functions are symmetric (they are real, and the second one is
a translate of the first by half a lattice vector).

However, closer inspection reveals that the optimized Wannier
functions are not exponentially localized, and in fact decay as
$1/r^{2}$ for large $r$ (Figure \ref{fig:oneoverrtwo}). The origin of
this slow decay of the Wannier functions is the kink that appears at
$\xi = \pm 3/2$ in their Fourier transform as in Figure \ref{fig:kink},
which also shows that the Wannier functions are compactly
supported on $[-3/2,3/2]$ in Fourier space. This is because only the
first three bands are occupied by the Wannier functions, as can be
seen via inspection of the $U$ matrix. Since band $3$ crosses with band
$4$ at $k=1/2$ (corresponding to basis functions $K = 1$ and $K=-2$),
the continuity of $\wt{\psi}_{nk}$ with respect to $k$ means that
$U_{31}(1/2) = U_{32}(1/2) = 0$. But the first derivative is not zero
at $1/2$, which in turns creates the kink in Fourier space at
$K=1,k=1/2$ and $K=-2,k=1/2$, that is, $\xi = \pm 3/2$.

\begin{figure}[h!]
  \centering
  \begin{subfigure}[t]{.45\textwidth}
    \includegraphics[width=\textwidth]{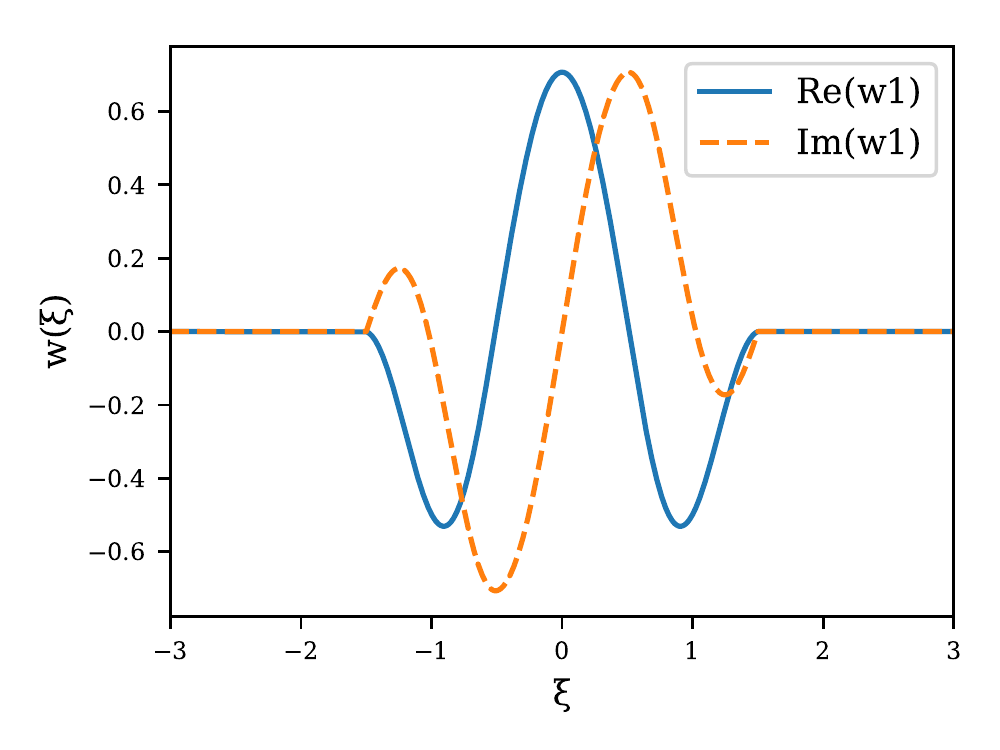}
    \caption{$w_{1}$ in Fourier space: a clear kink is
      visible at $\xi = \pm 3/2$}
  \label{fig:kink}
  \end{subfigure}
  \quad
  \begin{subfigure}[t]{.45\textwidth}
    \includegraphics[width=\textwidth]{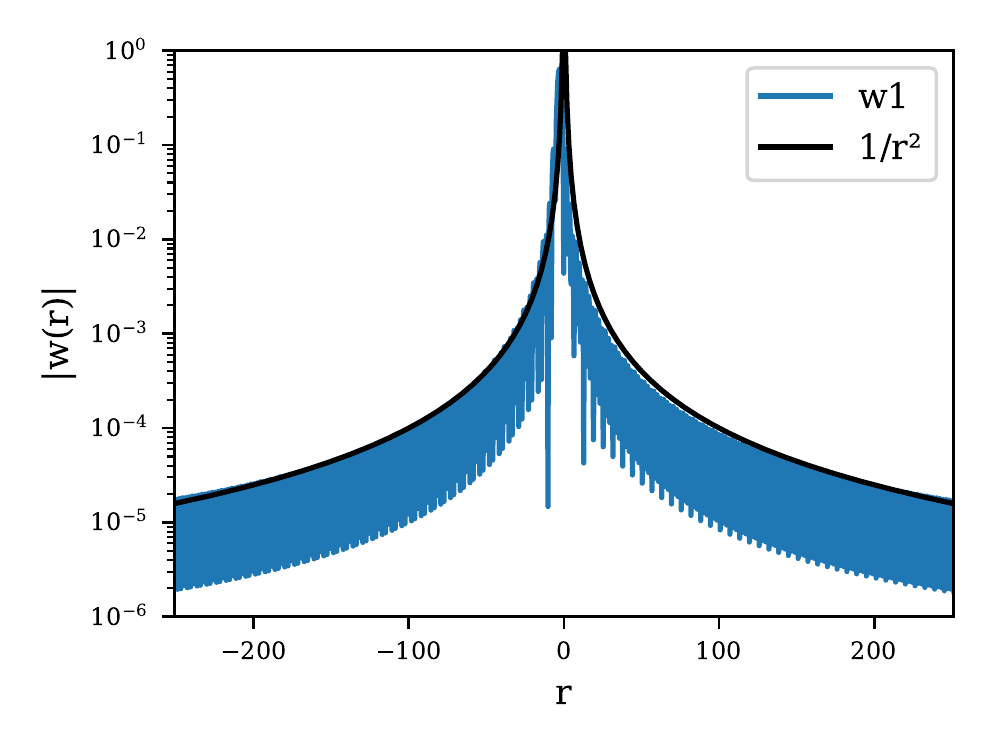}
    \caption{$w_{1}$ in real space, showing excellent agreement with a decay
      rate $r^{-2}$.}
    \label{fig:oneoverrtwo}
  \end{subfigure}
  \caption{Optimized Wannier functions obtained on a $k$-space grid of size $N=80$. The kink
    in Fourier space at $\xi = \pm 3/2$ translates to a $r^{-2}$ decay
    (only $w_{1}$ represented for clarity). The
    maximum reconstruction error on the first band is $3 \times
    10^{-5}$.}
  \label{fig:MLWF_1D}
\end{figure}

The fact that the optimized Wannier functions are only weakly localized is surprising
at first glance, because it was proven in \cite{PanatiPisante2013}
that, for isolated bands, maximally-localized Wannier functions are exponentially localized. What is different in this case? The crucial
point in the analysis of \cite{PanatiPisante2013} is that
the spread is similar to a Dirichlet energy, or a $H^{1}$ norm,
in $k$ space. Then the $\wt{\psi}_{nk}$ are shown to satisfy an elliptic
equation in $k$ space, which by a bootstrap argument implies their
analyticity. Here, this argument breaks down because the constraint
$P_{w}(k)P_{f}(k) = P_{f}(k)$ is discontinuous at crossings. This
creates an effective ``boundary condition'' for the gauge $U(k)$ at
crossings that destroys the regularity. A simple analogy is that
eigenfunctions of the Laplacian on $[0,1]$ (critical points of the
Dirichlet energy) are smooth on $\R$ when periodic boundary conditions
are imposed, but generically produce kinks at $0$ and $1$ when
Dirichlet boundary conditions are imposed. Here also, the effective
boundary condition $U_{31}(1/2)=U_{32}(1/2) = 0$ produces a kink.

To remedy this, we show on this particular example how to build
Wannier functions that are of class $C^{\infty}$ in Fourier space (and therefore
decay faster than any inverse polynomial in real space). 
In order to do so, we define the function $\alpha(x) = e^{-1/x}$ for
$x \geq 0$, $\alpha(x) = 0$ for $x < 0$. This function is
$C^{\infty}$, and identically zero for $x \leq 0$. The function
$f(x) = \frac{\alpha(x)}{\alpha(x)+\alpha(1-x)}$ is therefore
$C^{\infty}$ on $\mathbb R$, equal to $0$ for $x \leq 0$, and equal to
$1$ for $x \geq 1$.

Given $U(0^{+}), U(1/2^{\pm}), U(1^{-})$ obtained from the optimized
Wannier functions, we construct a new gauge
\begin{align*}
  \wt U(k) =
  \begin{cases}
    (1-f(2k)) U(0^{+}) + f(2k) U(1/2^{-}) & \text{ if $0 < k \le 1/2$,}\\
    (1-f(2k-1)) U(1/2^{+}) + f(2k-1) U(1^{-}) & \text{ if $1/2 < k < 1$.}
  \end{cases}
\end{align*}
This produces a new set of $\wt{\psi}_{nk}$ that are smooth with
respect to $k$, but not orthogonal and do not span the frozen bands.
We now impose those conditions in the same way we do for the SCDM
procedure as described in Section~\ref{sec:scdm}. In this specific example, this produces a
smooth gauge, as illustrated in Figure \ref{fig:smoothed_WF_1D}.

\begin{figure}[h!]
  \centering
  \begin{subfigure}[t]{.45\textwidth}
    \includegraphics[width=\textwidth]{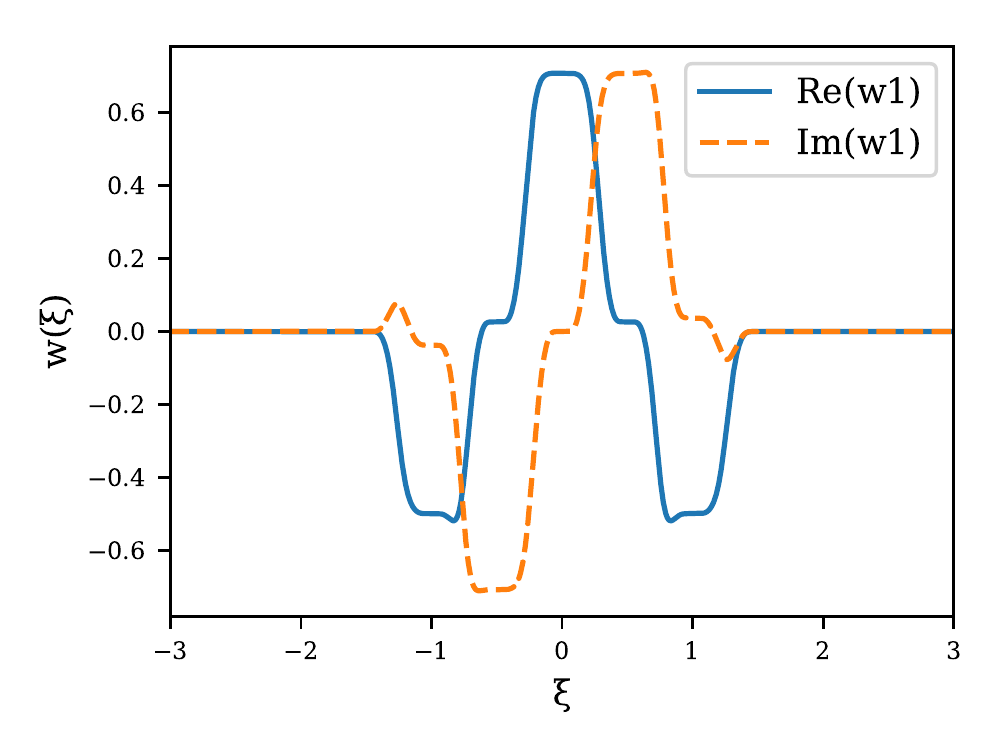}
    \caption{Smoothed Wannier function in Fourier space: although the
      variations are more rugged than Figure \ref{fig:MLWF_1D}, it is
      a $C^{\infty}$ function.}
  \end{subfigure}
  \quad
  \begin{subfigure}[t]{.45\textwidth}
    \includegraphics[width=\textwidth]{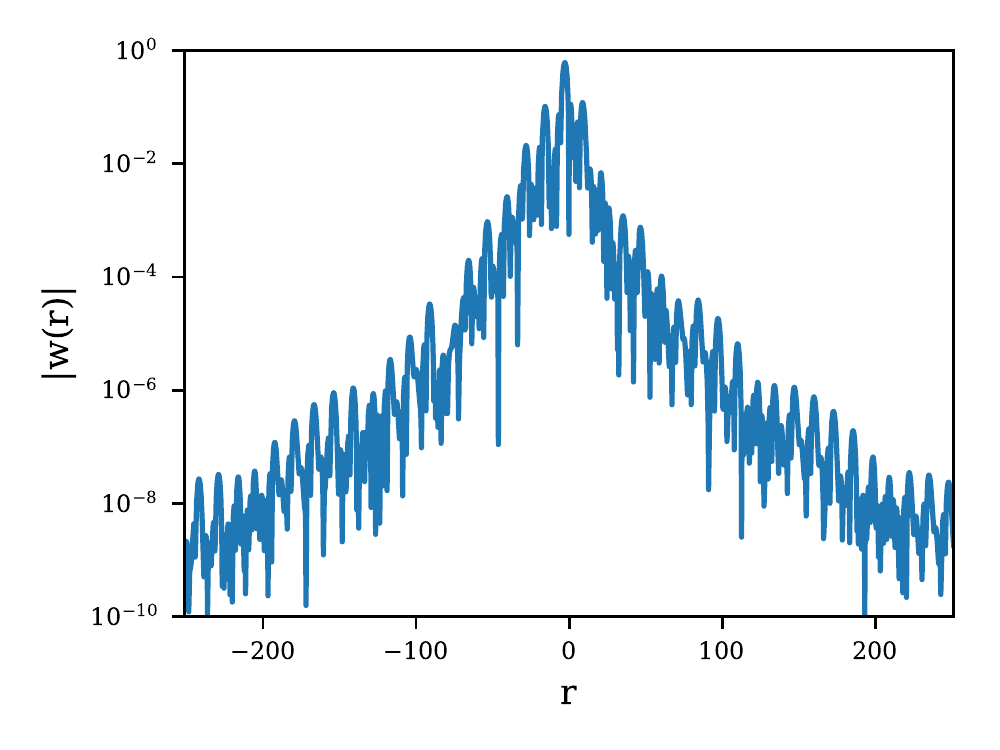}
    \caption{Smoothed Wannier function in real space: the asymptotic
      decay is much faster than in Figure \ref{fig:MLWF_1D}.}
  \end{subfigure}
  \caption{Smoothed Wannier functions obtained on a $k$-space grid of
    size $N=80$ (see main text for details, only the first Wannier
    function represented for clarity). The maximum
    reconstruction error on the first band is $2 \times 10^{-9}$. }
  \label{fig:smoothed_WF_1D}
\end{figure}

The Wannier functions obtained in this way display more rugged
variation in Fourier space (a consequence of the use of the function
$f$ above), and accordingly the decay is slower for small values of
$x$. However, because the gauge is smooth, the Wannier functions decay
much faster for large values. Since the gauge is of class $C^{\infty}$
but not analytic, the Wannier functions decay faster than any
polynomial, but not exponentially. Since the error of Wannier
interpolation on the first band is determined by the interaction of
the Wannier functions on the supercell with their periodic images, the
faster asymptotic decay leads to a better reconstruction of the first
band in this case. Numerical tests indicate that the cross-over point
(above which the reconstruction error on the first band with the
smoothing procedure is smaller than that with the optimized Wannier
functions) is around $N=12$.

\subsubsection{General case in 1D}
In the one-dimensional free electron gas crossings only happen at
half-integers, and between two bands at the same time. Accordingly,
whenever $N_{w} \geq N_{f}+1$ we can construct optimized Wannier functions similar to the
ones above\textemdash they are translates of each other, compactly
supported in Fourier space, and decay as $1/r^{2}$. When
$N_{w} = N_{f}$ (that is, we are treating a metal as if it was an
insulator), the gauge is discontinuous, and the corresponding Wannier
functions decay as $1/r$.

\subsection{The 2D case}
In 2D, the first four bands of the free electron gas are degenerate at $\vk = (1/2,1/2)$,
corresponding to the wave vectors 
\[
\vK = (0,0), (-1,0),(0,-1),(-1,-1).
\]
This means that $P_{w}(1/2,1/2)$ must span the four-dimensional
subspace corresponding to those four wave vectors. Therefore, freezing
the first band can only produce localized Wannier functions when
$N_{w} \geq 4$. Accordingly, we consider the case
$N_{f} = 1$ and $N_{w} = 4.$

Similarly to the 1D case, the optimized Wannier functions are real,
and differ from each other only by a change of origin. However, they
are not compactly supported in Fourier space, and instead have
decaying components on arbitrarily large wave vectors $\vK$ (see
Figure \ref{fig:2D_wan_abs}). Their support in Fourier space displays
a checkerboard pattern, and in particular has corners. Furthermore, we
observe numerically that the gradient in Fourier space blows up
near those corners when the size of the $k$-point grid $N$ is
increased (see Figure \ref{fig:2D_wan_grad}). The maximum value of the
gradient is found to behave as $N^{1/3}$. This is consistent with the
singularity of the first eigenfunction of the Laplace operator with
Dirichlet boundary conditions on a domain with corners, which behave
as $r^{2/3} \sin(2\theta/3)$ near a corner described by
$\theta \in [0,3\pi/2]$ in polar coordinates \cite{dauge2006elliptic}.
In particular, their gradient behaves as $r^{-1/3}$ near the corner,
explaining the divergence as $N^{1/3}$ when discretized on a grid. As
in the 1D case, imposing frozen bands acts as an effective boundary
condition for the gauge and destroys the regularity.

The decay of the optimized Wannier functions in Fourier space is
dictated by the singularity in Fourier space. More precisely, the
first derivative is discontinuous along edges, which corresponds to a
decay of the Wannier functions in the $x$ and $y$ directions as
$1/r^{2}$.

\begin{figure}[h!]
  \centering
  \begin{subfigure}[t]{.45\textwidth}
    \includegraphics[width=\textwidth]{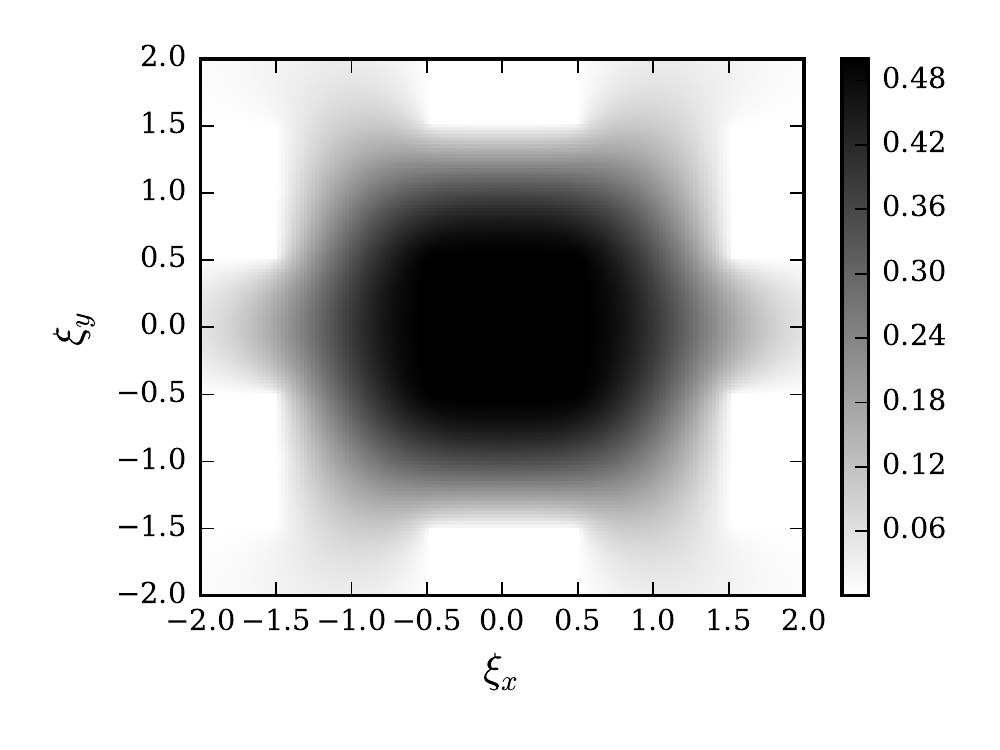}
    \caption{$|\widehat w_{2}(\xi)|$. The function has components on
      arbitrarily large wave vectors.}
    \label{fig:2D_wan_abs}
  \end{subfigure}
  \quad
  \begin{subfigure}[t]{.45\textwidth}
    \includegraphics[width=\textwidth]{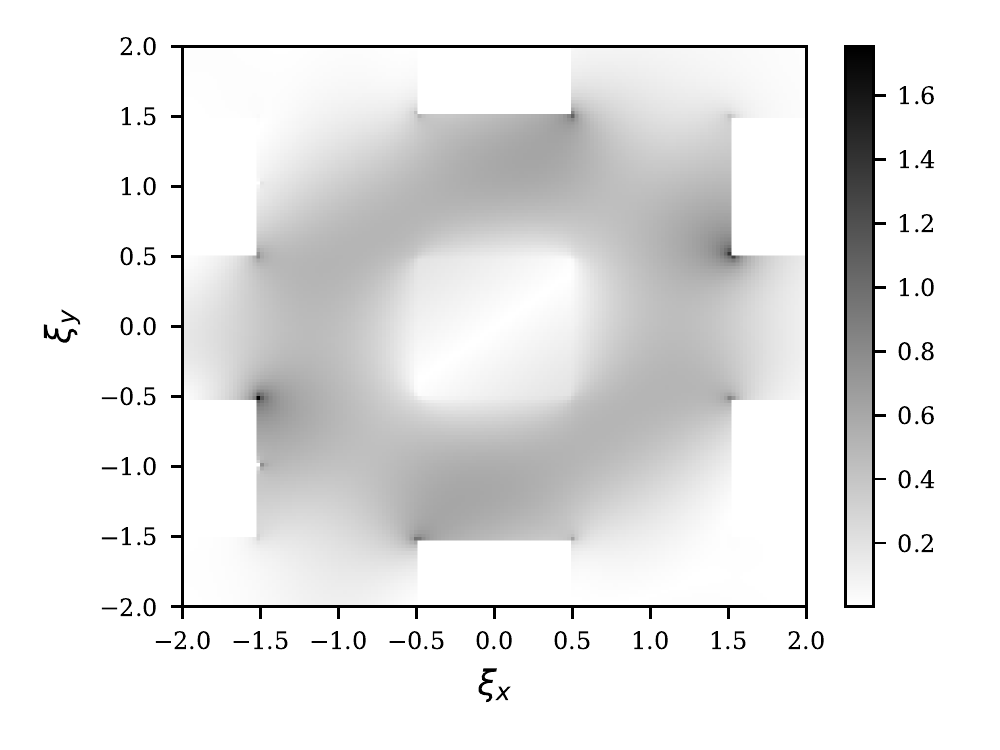}
    \caption{$| \nabla_{\xi}\widehat w_{2}(\xi)|$, clearly showing the
    divergence on corners and discontinuity on edges.}
    \label{fig:2D_wan_grad}
  \end{subfigure}

    \begin{subfigure}[t]{.45\textwidth}
      \includegraphics[width=\textwidth]{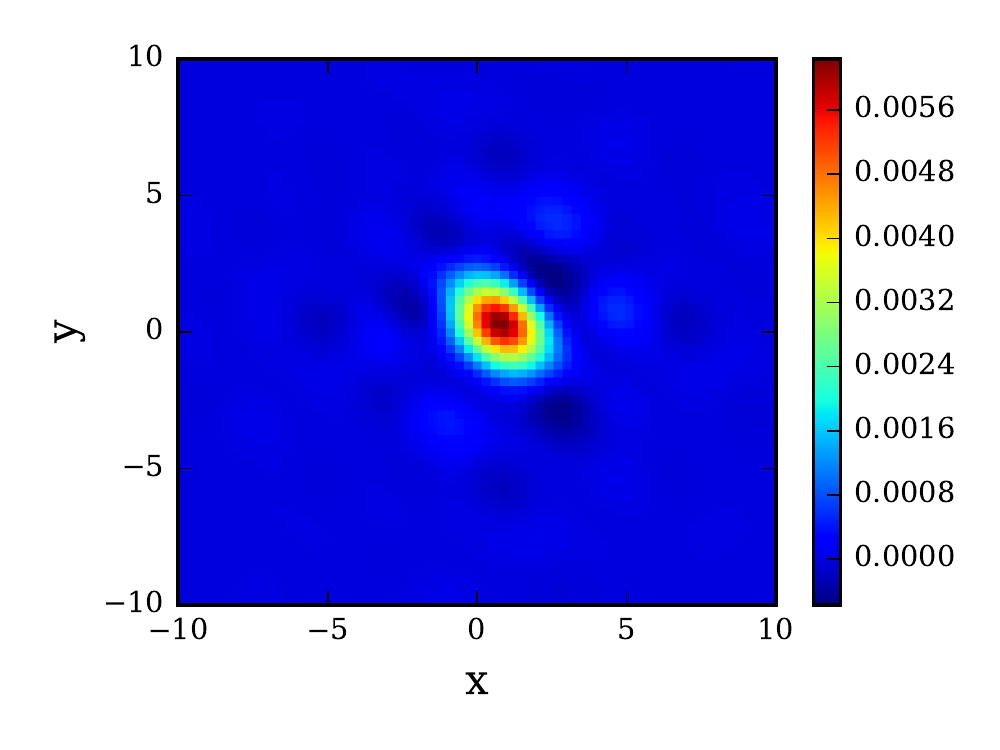}
    \caption{$w_{2}(\vr)$.}
  \end{subfigure}
  \quad
    \begin{subfigure}[t]{.45\textwidth}
      \includegraphics[width=\textwidth]{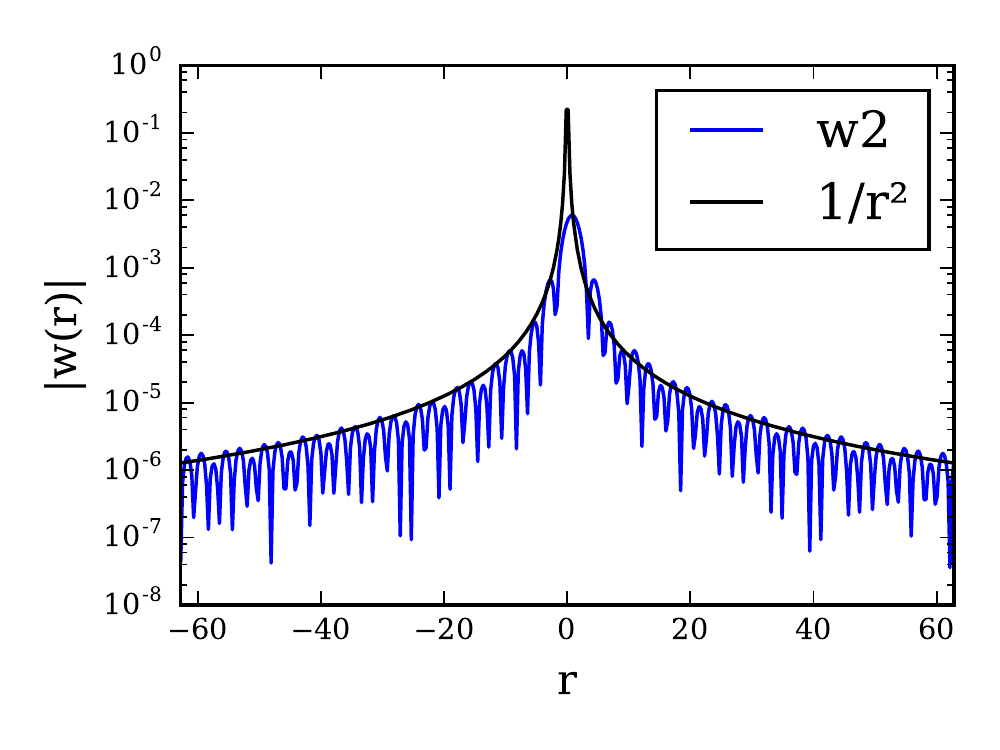}
    \caption{Slice of $w_{2}$ at $y=0$, showing the $1/r^{2}$ decay.}
  \end{subfigure}
  
  \caption{Second of four optimized Wannier functions for the free
    electron gas in 2D with one frozen band, on a
    $40\times40$ $k$-point grid. The other three Wannier functions are
    similar to $w_{2}$.}
  \label{fig:2D_wan}
\end{figure}

\subsection{Discussion}
These results show that the maximally-localized Wannier functions only
decay algebraically in general. Although numerical results are harder
to obtain for real materials, we expect our analysis to carry through:
an eigenvalue crossing at a particular $\vk$ point acts as a constraint
on the gauge, which must at that point be able to span the crossing
eigenspace for the gauge to be continuous. Minimizing the second
moments of the Wannier functions yields a gauge with a
square-integrable but discontinuous first derivative at the crossing
points, resulting in a weak localization.

Our findings are to be contrasted with the recent theoretical result
of \cite{CGLM}, which proves under generic hypotheses that there
exists almost-exponentially localized Wannier functions. This simply
means that, unlike in the case of insulators, for entangled band
structures minimizing the second moments is not an asymptotically
optimal strategy.

To get a faster asymptotic decay, one could minimize higher moments.
This corresponds to minimizing integrals of higher derivatives, which
have to be approximated by more complex stencils, and require the
computation of additional overlaps between the $u_{n\vk}$ than simply
the nearest neighbors. This becomes numerically expensive and complex
to implement. Another possibility is to perform local ``smoothing
surgeries'' similar to the one we demonstrated in one dimension. This
is likely to be useful only for very fine $\vk$-point grids however.

\section{Conclusion and discussion}\label{sec:conclusion}

We have developed a variational formulation that, paired with a
specific initialization strategy, is able to robustly construct
Wannier functions for systems with entangled band structure.
Importantly, the definition of Wannier functions must be generalized,
allowing them to lie in a subspace that contains, but is larger than,
the spectral subspace of interest. While this condition adds extra
constraints to our optimization problem, and they can be phrased in
many theoretically equivalent ways, we find one that is particularly
convenient. This results in a formulation that matches that of partly
occupied Wannier functions \cite{ThygesenHansenJacobsen2005}, and
allows us to view the widely used disentanglement procedure as an
alternating minimization algorithm\textemdash albeit one that only
takes a single alternation step. As the underlying problem is
non-convex, our choice of initialization strategy via the SCDM
methodology is key. As demonstrated with several real materials, our
method is robust and effective at finding localized functions and
enabling good quality band interpolation. Our variational formulation
is versatile, and can be modified relatively easily to accommodate
additional constraints, such as symmetry constraints, for certain type
of real materials. It would also be interesting to study the behavior
of localization properties of generalized Wannier functions for
systems with non-trivial topological characters.

We also study the free electron gas, providing interesting insights
into the further theoretical study of the localization properties of
generalized Wannier functions. We find that the minimization of the
second moments of the Wannier functions only imposes a relatively weak
algebraic decay. Our analysis suggests that, for real 3D materials,
the disentangled Wannier functions decay asymptotically slowly as
well. Further localization is possible, but the method we present here
is likely to only be useful for very fine $\vk$-point grids. The
computation of Wannier functions that are localized in both the
pre-asymptotic and asymptotic regime remains an interesting open question.

\section*{Acknowledgments}

This work was partially supported by the National Science Foundation
under Grant No. DMS-1652330, the Department of Energy under Grants No.
DE-SC0017867 and No. DE-AC02-05CH11231, and the SciDAC project (L. L.);
and the National Science Foundation under grant number DMS-1606277 (A.
D.).

\appendix

\section{Gradient of the objective function}\label{app:gradient}


If $F$ is a function from a complex Hilbert space to $\R$, we recall that its
gradient (sometimes written $\partial F / \partial z^{*}$) is defined as the (unique) vector $g$ such that
\begin{align*}
  F(z+h) = F(z) + \Re \lela g, h\rira + O(h^{2}).
\end{align*}
All the derivatives and gradient below are in this sense. Note that
this is different from the notion of derivative of a $\C \to \C$
function, which is not relevant here (non-trivial complex-to-real
functions are not complex differentiable). The advantage of this
definition is that it allows a straightforward translation of
first-order (but not second-order) optimization algorithms in complex
variables.


We use the numerical setup of \cite{MarzariVanderbilt1997}. Recall
that the $\vk$-point grid is discretized, with a total of $N_{\vk}$
points, and on each $\vk$ in that grid, we are given a set of
$N_{g} \times N_{o}$ matrices $u_{\vk}$ representing the $N_{o}$
orbitals in the outer window discretized on a space of dimension
$N_{g}$. The vectors $\{\vb\}$ are displacements from one $\vk$-point to a
set of neighbors, and $w_{\vb}$ are weights chosen to satisfy
\begin{align*}
  \sum_{\vb} w_{\vb}\vb\vb^{T} = I_{3},
\end{align*}
so that the gradient of a function $f(\vk)$ can be approximated
by
\begin{align*}
  \nabla f(\vk) \approx \sum_{\vb} w_{\vb} (f(\vk+\vb)-f(\vk))\vb
\end{align*}

We look for a set of $N_{o} \times N_{w}$ matrices $\{U_{\vk}\}$ with
orthogonal columns, which define Wannier functions by \eqref{eqn:rotateentangle}. Let
\begin{align*}
  M_{\vk,\vb,m,n} = \langle u_{m,\vk}, u_{n,\vk+\vb}\rangle
\end{align*}
be the $N_{o} \times N_{o}$ overlap matrix between the bands, which is
an input of the algorithm. Then the $N_{w} \times N_{w}$ overlap matrix between the Wannier
functions defined by $U(\vk)$ is
\begin{align*}
  N_{\vk,\vb} = U_{\vk}^{*}M_{\vk,\vb}U_{\vk+\vb}.
\end{align*}
The Marzari-Vanderbilt spread functional is given by
\begin{align*}
  \Omega &= \sum_{n} \lela |\vr|^{2}\rira_{n} - |\lela \vr\rira_{n}|^{2}, \text{ where}\\
  \lela \vr^{2} \rira_{n} &= \frac 1 {N_{\vk}}\sum_{\vk,\vb} w_{\vb}\left( 1 - |N_{\vk\vb nn}|^{2} + (\Im \ln N_{\vk \vb nn})^{2}\right)\\
  \lela \vr \rira_{n} &= - \frac 1 {N_{\vk}}\sum_{\vk,\vb} w_{\vb}\Im \ln N_{\vk \vb nn}\vb
\end{align*}
(equations (11), (31) and (32) of
\cite{MarzariVanderbilt1997}).

We need to compute $\Omega(U + \Delta U)$ to first
order in $\Delta U$, from which we will identify $\nabla \Omega$ by
$\Omega(U+\Delta U) - \Omega(U) = \Re \sum_{\vk}\Tr((\nabla
\Omega)_{\vk}^{*} (\Delta U)_{\vk}) + O(\Delta U^{2})$. 

We begin with $\sum_{n}\langle |\vr|^{2} \rangle_{n}$ and consider the
following quantity:
\begin{align*}
  I = \sum_{\vk\vb n} f(N_{\vk\vb nn})
\end{align*}
where $f: \C \to \R$. Then, using the fact that the set of vectors
$\vb$ is symmetric (contains $\vb$ as well as $-\vb$), and that
$w_{-\vb}=w_{\vb}$,
\begin{align*}
  \Delta I &= \Re \sum_{\vk\vb n} f''(N_{\vk\vb nn})^{*} \Delta N_{\vk\vb nn}\\
  &= \Re \sum_{\vk\vb n} w_{\vb}f'(N_{\vk\vb nn})^{*} (\Delta A_{\vk}^{*} M_{\vk\vb } A_{\vk+\vb})_{nn} + f'(N_{\vk\vb nn})(\Delta A_{\vk+\vb}^{*}M_{\vk\vb }^{*}A_{\vk})_{nn}\\
  &= \Re \sum_{\vk\vb n} w_{\vb}f'(N_{\vk\vb nn})^{*} (\Delta A_{\vk}^{*} M_{\vk\vb } A_{\vk+\vb})_{nn} + f'(N_{\vk-\vb,\vb,nn})(\Delta A_{\vk}^{*}M_{\vk-\vb,\vb}^{*}A_{\vk-\vb})_{nn}\\
  &= \Re \sum_{\vk\vb n} w_{\vb}f'(N_{\vk\vb nn})^{*} (\Delta A_{\vk}^{*} M_{\vk\vb } A_{\vk+\vb})_{nn} + f'(N_{\vk,-\vb,nn}^{*})(\Delta A_{\vk}^{*}M_{\vk,-\vb}^{*}A_{\vk-\vb})_{nn}\\
  &= \Re \sum_{\vk\vb n} w_{\vb}(f'(N_{\vk\vb nn})^{*} + f'(N_{\vk\vb nn}^{*})) (\Delta A_{\vk}^{*} M_{\vk\vb } A_{\vk+\vb})_{nn}
\end{align*}
and the gradient is therefore
\begin{align*}
  (\nabla I)_{\vk mn} &= \sum_{\vb} w_{\vb}(f'(N_{\vk\vb nn})^{*} + f'(N_{\vk\vb nn}^{*})) (M_{\vk\vb } A_{\vk+\vb})_{mn}
\end{align*}
It can be checked using similar arguments that, when
\begin{align*}
  I =\sum_{\vk\vb n} \vb
g(N_{\vk\vb nn}),
\end{align*}
then
\begin{align*}
  (\nabla I)_{\vk mn} &= -\sum_{\vb} w_{\vb}(g'(N_{\vk\vb nn})^{*} + g'(N_{\vk\vb nn}^{*})) (M_{\vk\vb } A_{\vk+\vb})_{mn}\vb
\end{align*}
Applying these formulas with
\begin{align*}
  f(z) &= 1-|z|^{2} + (\Im \ln z)^{2},\quad f'(z) = -2z + 2 \frac {i\Im \ln z} {z^{*}}\\
  g(z) &= -\Im \ln z, \quad g'(z) = -\frac{i}{z^{*}}
\end{align*}
we get
\begin{align*}
  (\nabla \Omega)_{\vk mn} &= \frac{4}{N_{\vk}}\sum_{\vb}
  \left(-N_{\vk\vb nn}^{*} - i\frac{\Im \ln N_{\vk\vb nn} + \lela r
  \rira_{n} \cdot \vb}{N_{\vk\vb nn}}\right)(M_{\vk\vb } A_{\vk+\vb})_{mn}
\end{align*}

Note that this is the \textit{unconstrained} gradient of $\Omega$ with
respect to $U$. Using the chain rule one with $U(\vk)=\begin{bmatrix}
  I_{N_{f}} & 0 \\
    0 & Y
  \end{bmatrix} X$ one can easily
derive the gradient of $\Omega$ with respect to $(X,Y)$. Then we
simply have to minimize with respect to $(X,Y)$ subject to the
orthogonality constraints for $X$ and $Y$ using standard methods \cite{edelman1998geometry,absil2009optimization}.



\bibliographystyle{siam}
\bibliography{wannier}

\begin{thebibliography}{10}

\bibitem{absil2009optimization}
{\sc P-A. Absil, R.~Mahony, and R.~Sepulchre}, {\em Optimization algorithms on
  matrix manifolds}, Princeton University Press, 2009.

\bibitem{bezanson2017julia}
{\sc J.~Bezanson, A.~Edelman, S.~Karpinski, and V.B. S.}, {\em Julia: A fresh
  approach to numerical computing}, SIAM Review, 59 (2017), pp.~65--98.

\bibitem{Blount1962}
{\sc E.~I. Blount}, {\em Formalisms of band theory}, Solid State Phys., 13
  (1962), pp.~305--373.

\bibitem{BowlerMiyazaki2012}
{\sc D.~R. Bowler and T.~Miyazaki}, {\em {O(N)} methods in electronic structure
  calculations}, Rep. Prog. Phys., 75 (2012), p.~036503.

\bibitem{BrouderPanatiCalandraEtAl2007}
{\sc C.~Brouder, G.~Panati, M.~Calandra, C.~Mourougane, and N.~Marzari}, {\em
  Exponential localization of {Wannier} functions in insulators}, Phys. Rev.
  Lett., 98 (2007), p.~046402.

\bibitem{CancesLevittPanatiEtAl2017}
{\sc E.~Canc{\`e}s, A.~Levitt, G.~Panati, and G.~Stoltz}, {\em Robust
  determination of maximally-localized {Wannier} functions}, Phys. Rev. B, 95
  (2017), p.~075114.

\bibitem{CGLM}
{\sc H.~Cornean, D.~Gontier, A.~Levitt, and D.~Monaco}, {\em Localised wannier
  functions in metallic systems}.
\newblock in preparation, 2017.

\bibitem{DamleLin2017}
{\sc A.~Damle and L.~Lin}, {\em Disentanglement via entanglement: A unified
  method for wannier localization}, arXiv:1703.06958,  (2017).

\bibitem{DamleLinYing2015}
{\sc A.~Damle, L.~Lin, and L.~Ying}, {\em Compressed representation of
  {K}ohn--{S}ham orbitals via selected columns of the density matrix}, J. Chem.
  Theory Comput., 11 (2015), pp.~1463--1469.

\bibitem{DamleLinYing2017a}
{\sc Anil Damle, Lin Lin, and Lexing Ying}, {\em Scdm-k: Localized orbitals for
  solids via selected columns of the density matrix}, J. Comput. Phys., 334
  (2017), pp.~1 -- 15.

\bibitem{dauge2006elliptic}
{\sc M.~Dauge}, {\em Elliptic boundary value problems on corner domains:
  smoothness and asymptotics of solutions}, vol.~1341, Springer, 2006.

\bibitem{ELiLu2010}
{\sc W.~E, T.~Li, and J.~Lu}, {\em Localized bases of eigensubspaces and
  operator compression}, Proc. Natl. Acad. Sci., 107 (2010), pp.~1273--1278.

\bibitem{edelman1998geometry}
{\sc A.~Edelman, T.A. Arias, and S.T. Smith}, {\em The geometry of algorithms
  with orthogonality constraints}, SIAM journal on Matrix Analysis and
  Applications, 20 (1998), pp.~303--353.

\bibitem{FosterBoys1960}
{\sc J.~M. Foster and S.~F. Boys}, {\em Canonical configurational interaction
  procedure}, Rev. Mod. Phys., 32 (1960), p.~300.

\bibitem{GiannozziBaroniBoniniEtAl2009}
{\sc Paolo Giannozzi, Stefano Baroni, Nicola Bonini, Matteo Calandra, Roberto
  Car, Carlo Cavazzoni, Davide Ceresoli, Guido~L Chiarotti, Matteo Cococcioni,
  Ismaila Dabo, Andrea~Dal Corso, Stefano de~Gironcoli, Stefano Fabris, Guido
  Fratesi, Ralph Gebauer, Uwe Gerstmann, Christos Gougoussis, Anton Kokalj,
  Michele Lazzeri, Layla Martin-Samos, Nicola Marzari, Francesco Mauri,
  Riccardo Mazzarello, Stefano Paolini, Alfredo Pasquarello, Lorenzo Paulatto,
  Carlo Sbraccia, Sandro Scandolo, Gabriele Sclauzero, Ari~P Seitsonen,
  Alexander Smogunov, Paolo Umari, and Renata~M Wentzcovitch}, {\em {QUANTUM
  ESPRESSO: a modular and open-source software project for quantum simulations
  of materials}}, J. Phys.: Condens. Matter, 21 (2009), pp.~395502--395520.

\bibitem{Goedecker1999}
{\sc S.~Goedecker}, {\em {Linear scaling electronic structure methods}}, Rev.
  Mod. Phys., 71 (1999), pp.~1085--1123.

\bibitem{Gygi2009}
{\sc F.~Gygi}, {\em Compact representations of {K}ohn--{S}ham invariant
  subspaces}, Phys. Rev. Lett., 102 (2009), p.~166406.

\bibitem{hager2005new}
{\sc William~W Hager and Hongchao Zhang}, {\em A new conjugate gradient method
  with guaranteed descent and an efficient line search}, SIAM Journal on
  optimization, 16 (2005), pp.~170--192.

\bibitem{HohenbergKohn1964}
{\sc P.~Hohenberg and W.~Kohn}, {\em {Inhomogeneous electron gas}}, Phys. Rev.,
  136 (1964), pp.~B864--B871.

\bibitem{KochGoedecker2001}
{\sc E.~Koch and S.~Goedecker}, {\em Locality properties and {Wannier}
  functions for interacting systems}, Solid State Commun., 119 (2001), p.~105.

\bibitem{Kohn1959}
{\sc W.~Kohn}, {\em Analytic properties of {Bloch} waves and {Wannier}
  functions}, Phys. Rev., 115 (1959), p.~809.

\bibitem{Kohn1996}
\leavevmode\vrule height 2pt depth -1.6pt width 23pt, {\em Density functional
  and density matrix method scaling linearly with the number of atoms}, Phys.
  Rev. Lett., 76 (1996), pp.~3168--3171.

\bibitem{KohnSham1965}
{\sc W.~Kohn and L.~Sham}, {\em {Self-consistent equations including exchange
  and correlation effects}}, Phys. Rev., 140 (1965), pp.~A1133--A1138.

\bibitem{kuchment2016overview}
{\sc P.~Kuchment}, {\em An overview of periodic elliptic operators}, Bulletin
  of the American Mathematical Society, 53 (2016), pp.~343--414.

\bibitem{Loewdin1950}
{\sc P.-O. L{\"o}wdin}, {\em On the non-orthogonality problem connected with
  the use of atomic wave functions in the theory of molecules and crystals}, J.
  Chem. Phys., 18 (1950), pp.~365--375.

\bibitem{MarzariMostofiYatesEtAl2012}
{\sc N.~Marzari, A.~A. Mostofi, J.~R. Yates, I.~Souza, and D.~Vanderbilt}, {\em
  Maximally localized {W}annier functions: Theory and applications}, Rev. Mod.
  Phys., 84 (2012), pp.~1419--1475.

\bibitem{MarzariVanderbilt1997}
{\sc N.~Marzari and D.~Vanderbilt}, {\em Maximally localized generalized
  {W}annier functions for composite energy bands}, Phys. Rev. B, 56 (1997),
  p.~12847.

\bibitem{MostofiYatesLeeEtAl2008}
{\sc A.A. Mostofi, J.R. Yates, Y.S. Lee, I.~Souza, D.~Vanderbilt, and
  N.~Marzari}, {\em {Wannier90: a tool for obtaining maximally-localised
  {W}annier functions}}, Comput. Phys. Commun., 178 (2008), p.~685.

\bibitem{mostofi2014updated}
{\sc A.~Mostofi, J.R. Yates, G.~Pizzi, Y.~Lee, I.~Souza, D.~Vanderbilt, and
  N.~Marzari}, {\em An updated version of wannier90: A tool for obtaining
  maximally-localised {W}annier functions}, Computer Physics Communications,
  185 (2014), pp.~2309--2310.

\bibitem{MustafaCohCohenEtAl2015}
{\sc J.~I. Mustafa, S.~Coh, M.~L. Cohen, and S.~G. Louie}, {\em {Automated
  construction of maximally localized Wannier functions: Optimized projection
  functions method}}, Phys. Rev. B, 92 (2015), p.~165134.

\bibitem{Nenciu1991}
{\sc G.~Nenciu}, {\em Dynamics of band electrons in electric and magnetic
  fields: rigorous justification of the effective hamiltonians}, Rev. Mod.
  Phys., 63 (1991), pp.~91--127.

\bibitem{nocedal2006numerical}
{\sc J.~Nocedal and S.J. Wright}, {\em Numerical optimization}, Springer, 2006.

\bibitem{OzolinsLaiCaflischEtAl2013}
{\sc V.~Ozoli{\c{n}}{\v{s}}, R.~Lai, R.~Caflisch, and S.~Osher}, {\em
  Compressed modes for variational problems in mathematics and physics}, Proc.
  Natl. Acad. Sci., 110 (2013), pp.~18368--18373.

\bibitem{panati2007triviality}
{\sc G.~Panati}, {\em Triviality of {B}loch and {B}loch--{D}irac bundles}, in
  Annales Henri Poincar{\'e}, vol.~8, Springer, 2007, pp.~995--1011.

\bibitem{PanatiPisante2013}
{\sc G.~Panati and A.~Pisante}, {\em Bloch bundles, {Marzari-Vanderbilt}
  functional and maximally localized {Wannier} functions}, Commun. Math. Phys.,
  322 (2013), pp.~835--875.

\bibitem{SouzaMarzariVanderbilt2001}
{\sc I.~Souza, N.~Marzari, and D.~Vanderbilt}, {\em Maximally localized wannier
  functions for entangled energy bands}, Phys. Rev. B, 65 (2001), p.~035109.

\bibitem{ThygesenHansenJacobsen2005}
{\sc K.~S. Thygesen, L.~B. Hansen, and K.~W. Jacobsen}, {\em Partly occupied
  wannier functions}, Phys. Rev. Lett., 94 (2005), p.~026405.

\bibitem{von1993merkwurdige}
{\sc J.~von Neumann and E.P. Wigner}, {\em {\"U}ber merkw{\"u}rdige diskrete
  eigenwerte}, in The Collected Works of Eugene Paul Wigner, Springer, 1993,
  pp.~291--293.

\bibitem{Wannier1937}
{\sc G.~H. Wannier}, {\em The structure of electronic excitation levels in
  insulating crystals}, Phys. Rev., 52 (1937), p.~191.

\bibitem{YatesWangVanderbiltEtAl2007}
{\sc J.~R. Yates, X.~Wang, D.~Vanderbilt, and I.~Souza}, {\em Spectral and
  {Fermi} surface properties from {Wannier} interpolation}, Phys. Rev. B, 75
  (2007), p.~195121.

\end{thebibliography}

\end{document}